\documentclass[preprint]{elsarticle}

\usepackage{todonotes}

\usepackage{proof}

\usepackage{diagrams}

 \usepackage{tikz}
 \usetikzlibrary{trees}

 \usepackage{amssymb}
 \usepackage{amsfonts,amsmath,latexsym}

 \newcommand{\At}{\mathrm{\textrm{At}}}

 \newcommand{\ls}{\mathcal{L}_{\Sigma}}

\newcounter{proposition}
 \newtheorem{definition}{Definition}

 \newtheorem{corollary}{Corollary}
 \newtheorem{example}{Example}
 
\newtheorem{proposition}{Proposition}
\newtheorem{theorem}{Theorem}

\newproof{proof}{Proof}

\begin{document}

\begin{frontmatter}

 \author{Ekaterina Komendantskaya\corref{cor1}}
\ead{ek19@hw.ac.uk}
  \address{Department of Computer Science, Heriot-Watt University, Edinburgh, UK}

\author{John Power}
\ead{A.J.Power@bath.ac.uk}
\address{Department of Computer Science,
 University of Bath, BA2 7AY, UK}

\cortext[cor1]{Corresponding author}

 \title{Logic programming: laxness and saturation\tnoteref{t1}}
\tnotetext[t1]{No data was generated in the course of this research.}

 \newcounter{private}



 \begin{abstract}
A propositional logic program $P$ may be identified with a
$P_fP_f$-coalgebra on the set of atomic propositions in the
program. The corresponding $C(P_fP_f)$-coalgebra, where $C(P_fP_f)$ is
the cofree comonad on $P_fP_f$, describes derivations by resolution.
That correspondence has been developed to model first-order programs in two ways, 
with lax semantics and saturated semantics, based on
locally ordered categories and right Kan extensions respectively. We unify the two
approaches, exhibiting them as complementary rather than competing,
reflecting the theorem-proving and proof-search aspects of logic programming.
While maintaining that unity, we further refine lax semantics to give finitary models of logic progams with existential variables, and to develop a precise semantic relationship between variables in logic programming and worlds in local state. 
\end{abstract}


\end{frontmatter}

 \section{Introduction}

Over recent years, there has been a surge of interest in category theoretic semantics of logic programming. Research has focused on two ideas: lax semantics, proposed by the current authors and collaborators~\cite{KoPS}, and saturated semantics, proposed by Bonchi and Zanasi~\cite{BonchiZ15}. 
Both ideas are based on coalgebra, agreeing on variable-free logic programs. Both ideas use subtle, well-established category theory,  associated with locally ordered categories and with right Kan extensions respectively~\cite{K}. And both elegantly clarify and extend established logic programming constructs and traditions, for instance~\cite{GC} and~\cite{BMR}.

Until now, the two ideas have been presented as alternatives, competing with each other rather than complementing each other. A central thesis of this paper is that the competition is illusory, the two ideas being two views of a single, elegant body of theory, those views reflecting  different but complementary aspects of logic programming, those aspects  broadly corresponding with the notions of theorem proving and proof search. Such reconciliation has substantial consequences. In particular, it means that whenever one further refines one approach,
as we shall do to the original lax approach in two substantial ways here, one should test whether the proposed refinement also applies to the other approach, and see what consequences it has from the latter perspective. 

The category theoretic basis for both lax and saturated semantics is as follows. It has long been observed, e.g., in~\cite{BM,CLM}, that logic programs induce coalgebras, allowing coalgebraic modelling of their operational semantics. Using the definition of logic program in Lloyd's book~\cite{Llo}, given a set of atoms $At$, one can identify a variable-free logic program $P$ built over $At$ with a
$P_fP_f$-coalgebra structure on $At$, where $P_f$ is the finite
powerset functor on $Set$: each atom is the head of finitely many
clauses in $P$, and the body of each clause contains finitely many
atoms. It was shown in~\cite{KMP} that if $C(P_fP_f)$ is the cofree comonad
on $P_fP_f$, then, given a logic program $P$ qua $P_fP_f$-coalgebra,
the corresponding $C(P_fP_f)$-coalgebra structure characterises the
and-or derivation trees generated by $P$, cf.~\cite{GC}. That fact has formed the basis for our work on lax semantics~\cite{KoPS,KSH14,JohannKK15,FK15,FKSP16} and for Bonchi and Zanasi's work on saturation semantics~\cite{BZ,BonchiZ15}.
 
In attempting to extend the analysis to arbitrary logic programs, both groups followed the tradition of~\cite{ALM,BM,BMR,KP96}: given a signature $\Sigma$
of function symbols, let $\mathcal{L}_{\Sigma}$ denote the Lawvere
theory generated by $\Sigma$, and, given a logic program $P$ with
function symbols in $\Sigma$, consider the functor category
$[\mathcal{L}_{\Sigma}^{op},Set]$, extending the set $At$ of atoms in
a variable-free logic program to the functor from
$\mathcal{L}_{\Sigma}^{op}$ to $Set$ sending a natural number $n$ to
the set $At(n)$ of atomic formulae with at most $n$ variables
generated by the function symbols in $\Sigma$ and the predicate
symbols in $P$. We all sought to model $P$ by a
$[\ls^{op},P_fP_f]$-coalgebra $p:At\longrightarrow P_fP_fAt$ that, at
$n$, takes an atomic formula $A(x_1,\ldots ,x_n)$ with at most $n$
variables, considers all substitutions of clauses in $P$ into clauses
with variables among $x_1,\ldots ,x_n$ whose head agrees with
$A(x_1,\ldots ,x_n)$, and gives the set of sets of atomic formulae in
antecedents, naturally extending the construction for variable-free logic
programs.  However, that idea is too simple for two reasons. We all dealt with the
second problem in the same way, so we shall discuss it later, but the first problem is illustrated by  the following example.

\begin{example} \label{ex:listnat}
 ListNat (for lists of natural numbers) denotes the logic program\\
 $1.\  \mathtt{nat(0)}  \gets $\\
 $2.\  \mathtt{nat(s(x))} \gets  \mathtt{nat(x)}$\\
$3.\  \mathtt{list(nil)}  \gets $ \\
$4.\  \mathtt{list(cons (x, y))}  \gets  \mathtt{nat(x), list(y)}$\\

\noindent
ListNat has nullary function symbols $\mathtt{0}$ and $\mathtt{nil}$, a unary function
symbol $\mathtt{s}$, and a binary function symbol $\mathtt{cons}$. So the signature $\Sigma$ of ListNat contains four elements. 

There
 is a map in $\mathcal{L}_{\Sigma}$ of the form $0\rightarrow 1$ that
 models the nullary function symbol $0$. So, naturality of the map
 $p:At\longrightarrow P_fP_fAt$ in $[\mathcal{L}_{\Sigma}^{op},Set]$
 would yield commutativity of the diagram
\begin{diagram}
At(1) & \rTo^{p_1} & P_fP_fAt(1) \\
\dTo<{At(\mathsf{0})} & & \dTo>{P_fP_fAt(\mathsf{0})} \\
At(0) & \rTo_{p_0} & P_fP_fAt(0)
\end{diagram}
But consider $\mathtt{nat(x)}\in At(1)$: there is no clause of the
form $\mathtt{nat(x)}\gets \, $ in ListNat, so commutativity of the
diagram would imply that there cannot be a clause in ListNat of the
form $\mathtt{nat(0)}\gets \, $ either, but in fact there is one. Thus $p$ is
not a map in the functor category $[\mathcal{L}_{\Sigma}^{op},Set]$.
\end{example}

Proposed resolutions diverged: at CALCO in 2011, we
proposed lax transformations~\cite{KoP}, then at
CALCO 2013, Bonchi and Zanasi proposed saturation
semantics~\cite{BZ}. First we shall describe our approach.

Our approach was to relax the
naturality condition on $p$ to a subset condition, following~\cite{Ben,BKP,Kelly}, so that, given a map in
$\mathcal{L}_{\Sigma}$ of the form $f:n \rightarrow m$, the diagram
\begin{diagram}
At(m) & \rTo^{p_m} & P_fP_fAt(m) \\
\dTo<{At(f)} & & \dTo>{P_fP_fAt(f)} \\
At(n) & \rTo_{p_n} & P_fP_fAt(n)
\end{diagram}
 need not commute, but rather the composite via $P_fP_fAt(m)$ need
 only yield a subset of that via $At(n)$. So, for example,
 $p_1(\mathtt{nat(x)})$ could be the empty set while
 $p_0(\mathtt{nat(0)})$ could be non-empty in the semantics for
 ListNat as required. We extended $Set$ to $Poset$ in order to express
 such laxness, and we adopted established category theoretic research
 on laxness, notably that of~\cite{Kelly}, in order to prove that a
 cofree comonad exists and, on programs such as ListNat, behaves as we wish.
 This agrees with, and is indeed an instance of, He Jifeng and Tony Hoare's use of laxness to model data refinement~\cite{HH,HH1,KP1,P}.

Bonchi and Zanasi's approach was to use saturation semantics~\cite{BZ,BonchiZ15},
following~\cite{BM}. The key category theoretic result that supports it asserts that, regarding $ob(\ls)$, equally $ob(\ls)^{op}$, as a discrete category with inclusion
functor $I:ob(\ls)\longrightarrow \ls$, the functor
\[
[I,Set]:[\ls^{op},Set]\longrightarrow [ob(\ls)^{op},Set]
\]
that sends $H:\ls^{op}\longrightarrow Set$ to the composite
$HI:ob(\ls)^{op} \longrightarrow Set$ has a right
adjoint, given by right Kan extension. The data for \mbox{$p:At\longrightarrow P_fP_fAt$}, although not forming a map in $[\ls^{op},Set]$, may be seen as a map in $[ob(\ls)^{op},Set]$. So, by the adjointness, the data for $p$ corresponds to a map $\bar{p}:At\longrightarrow R(P_fP_fAtI)$ in $[\ls^{op},Set]$, thus to a coalgebra on $At$ in $[\ls^{op},Set]$, where $R(P_fP_fAtI)$ is the right Kan extension of $P_fP_fAtI$ along the inclusion $I$. The right Kan extension is defined by
\[
R(P_fP_fAtI)(n) =  \prod_{m \in \ls} (P_fP_fAt(m))^{\ls(m,n)}
\]
and the function
\[
\bar{p}(n):At(n) \longrightarrow \prod_{m \in \ls} (P_fP_fAt(m))^{\ls(m,n)}
\]
takes an atomic formula $A(x_1,\ldots,x_n)$, and, for every substitution for $x_1,\ldots,x_n$ generated by the signature $\Sigma$, gives the set of sets of atomic formulae in the tails of clauses with head $A(t_1,\ldots,t_n)$, where the $t_i$'s are determined by the substitution. By
construction, $\bar{p}$  is natural, but one quantifies over all possible substitutions for $x_1,\ldots,x_n$ in order to obtain that naturality, and one ignores the laxness of $p$.

As we shall show in Section~\ref{sec:sat}, the two approaches can be unified. If one replaces
\[
[I,Set]:[\ls^{op},Set]\longrightarrow [ob(\ls)^{op},Set]
\]
by the inclusion
\[
[\ls^{op},Poset]\longrightarrow Lax(\ls^{op},Poset)
\]
$[\ls^{op},Set]$ being a full subcategory of $[\ls^{op},Poset]$, one obtains exactly Bonchi and Zanasi's correspondence between $p$ and $\bar{p}$, with exactly the same formula, starting from lax transformations as we proposed. Thus, from a category theoretic perspective, saturation can be seen as complementary to laxness rather than as an alternative to it. This provides a robustness test for future refinements to models of logic programming: a refinement of one view of category theoretic semantics can be tested by its effect on the other. We now turn to such refinements. 

Recently, we have refined lax semantics in two substantial ways, the first of which was the focus of the workshop paper~\cite{KoP1} that this paper extends, and the second of which we introduce here. For the first, a central contribution of lax semantics has been the inspiration
it provided towards the development of an efficient logic programming algorithm~\cite{KoPS,KSH14,JohannKK15,FK15,FKSP16}. That development drew our attention to the semantic significance of \emph{existential} variables: such variables do not appear in ListNat, and they are not needed for a considerable body of logic programming, but they do
appear in logic programs such as the following, which is a leading example in Sterling and Shapiro's book~\cite{SS}.

\begin{example} \label{ex:lp}
 GC (for graph connectivity) denotes the logic program\\
 $1.\  \mathtt{connected(x,x)}  \gets $\\
 $2.\  \mathtt{connected(x,y)}  \gets  \mathtt{edge(x,z)}, \mathtt{connected(z,y)}$\\
\end{example}
There is a variable $z$ in the tail of the second clause of GC that does not
appear in its head, whereas no such variable appears in ListNat. Such a variable is called an existential variable, the presence of which challenges the algorithmic significance of lax semantics. In describing the putative coalgebra
$p:At\longrightarrow P_fP_fAt$ just before Example~\ref{ex:listnat}, we referred
to \emph{all} substitutions of clauses in $P$ into clauses
with variables among $x_1,\ldots ,x_n$ whose head agrees with
$A(x_1,\ldots ,x_n)$. If there are no existential variables, that amounts to
term-matching, which is algorithmically efficient; but if existential
variables do appear, the mere presence of a unary function symbol generates an infinity of such
substitutions, creating algorithmic difficulty, which, when first introducing lax semantics, we,
also Bonchi and Zanasi,
avoided modelling by replacing the outer instance of $P_f$ by $P_c$, thus allowing for countably many
choices. That is the second of the two problems mentioned just before Example~\ref{ex:listnat}. We have long sought a more elegant resolution to that, one that restricts the construction of $p$ to finitely many substitutions. We finally found and presented such a resolution in the workshop paper~\cite{KoP1} that this paper extends. We  both refine it a little more, as explained later, and give more detail here. 

The conceptual key to the resolution was to isolate and give finitary lax semantics to the notion of \emph{coinductive tree}~\cite{KoP2,KoPS}. Coinductive trees arise from term-matching resolution~\cite{KoP2,KoPS}, which is a restriction of SLD-resolution. It
captures the theorem proving aspect of logic programming, which is distinct from,
but complementary with, its problem solving aspect, which is captured by SLD-resolution
~\cite{FK15,JohannKK15}. We called the derivation trees arising from
term-matching coinductive trees in order to mark their connection with coalgebraic logic
programming, which we also developed. Syntactically, one can observe the difference
between lax semantics and saturation semantics in that lax semantics models
coinductive trees, which are finitely branching, whereas saturation involves infinitely
many possible substitutions, leading Bonchi and Zanasi to model different kinds of trees,
their focus being on proof search rather than on theorem proving.

Chronologically,
we introduced lax semantics in 2011 as above~\cite{KoP}; lax semantics inspired us to investigate term-matching and to introduce the notion of coinductive tree~\cite{KoP2}; because of the possibility of existential variables, our lax semantics for coinductive trees, despite inspiring
the notion, was potentially infinitary~\cite{KoPS}; so we have now refined lax semantics to ensure finitariness of the semantics for coinductive trees, even in the presence of existential variables~\cite{KoP1}, introducing it in the workshop paper that this paper extends. We further refine lax semantics here to start to build a precise relationship with the semantics
of local variables~\cite{PP2}, which we plan to develop further in future. We regard
it as positive that lax semantics brings to the fore, in semantic terms, the significance of existential variables, and allows a precise semantic relationship between the role of variables
in logic programming and local variables as they arise in programming more generally.


The paper is organised as follows. In Section~\ref{sec:backr}, we set logic programming
terminology, explain the relationship between term-rewriting and SLD-resolution, and introduce
the notion of coinductive tree. In Section~\ref{sec:parallel}, we give semantics
for variable-free logic programs. This semantics could equally be seen as lax semantics or
saturated semantics, as they agree in the absence of variables. In Section~\ref{sec:recall},
we model coinductive trees for logic programs without existential variables and explain
the difficulty in modelling coinductive trees for arbitrary logic programs. In Section~\ref{sec:sat},
we recall saturation semantics and make precise the relationship between it and lax semantics.
We devote Section~\ref{sec:derivation} of the paper to refining lax semantics,
while maintaining the relationship with saturation semantics, to model the coinductive trees
generated by logic programs with existential variables, and in Section~\ref{sec:local}, we start to build a precise relationship with the semantics of local state~\cite{PP2}.

\section{Theorem proving in logic programming}\label{sec:backr}

 A \emph{signature} $\Sigma$ consists of a set $\mathcal{F}$ of
 function symbols $f,g, \ldots$ each equipped with an arity.  Nullary
 (0-ary) function symbols are constants. For any set $\mathit{Var}$ of
 variables, the set $Ter(\Sigma)$ of terms over $\Sigma$ is defined
 inductively as usual:
 \begin{itemize}
 \item $x \in Ter(\Sigma)$ for every $x \in \mathit{Var}$.
 \item If $f$ is an n-ary function symbol ($n\geq 0$) and $t_1,\ldots
   ,t_n \in Ter(\Sigma) $, then $f(t_1,\ldots
   ,t_n) \in Ter(\Sigma)$.  
 \end{itemize}

A \emph{substitution} over $\Sigma$ is a (total) function $\sigma:
\mathit{Var} \to \mathbf{Term}(\Sigma)$.  Substitutions are extended
from variables to terms as usual: if $t\in \mathbf{Term}(\Sigma)$ and
$\sigma$ is a substitution, then the {\em application} $\sigma(t)$ is
a result of applying $\sigma$ to all variables in $t$.  A substitution
$\sigma$ is a \emph{unifier} for $t, u$ if $\sigma(t) = \sigma(u)$,
and is a \emph{matcher} for $t$ against $u$ if $\sigma(t) = u$.  A
substitution $\sigma$ is a {\em most general unifier} ({\em mgu}) for
$t$ and $u$ if it is a unifier for $t$ and $u$ and is more general
than any other such unifier. A {\em most general matcher} ({\em mgm})
$\sigma$ for $t$ against $u$ is defined analogously.

In line with logic programming (LP) tradition~\cite{Llo}, we consider
a set $\mathcal{P}$ of predicate symbols each equipped with an arity.
It is possible to define logic programs over terms only, in line with the
term-rewriting (TRS) tradition~\cite{Terese}, as
in~\cite{JohannKK15}, but we will follow the usual LP tradition here.
That gives us the following inductive definitions of the sets of
atomic formulae, Horn clauses and logic programs (we also include the
definition of terms for convenience).

\begin{definition}\label{df:syntax}

\

Terms $Ter \ ::= \ Var \ | \ \mathcal{F}(Ter,..., Ter)$

   Atomic formulae (or atoms) $At \ ::= \ \mathcal{P}(Ter,...,Ter)$

  (Horn) clauses $HC \ ::= \ At \gets At,..., At$
	
	Logic programs $Prog \ ::= HC, ... , HC$
\end{definition}

In what follows, we will use letters $A,B,C,D$, possibly with
subscripts, to refer to elements of $At$.

Given a logic program $P$, we may ask whether a given atom is
logically entailed by $P$. E.g., given the program ListNat we may ask
whether $\mathtt{list(cons(0,nil))}$ is entailed by ListNat. The
following rule, which is a restricted form of SLD-resolution, provides
a semi-decision procedure to derive the entailment.

\begin{definition}[Term-matching (TM) Resolution]\label{def:resolution}
{\small
\[\begin{array}{c}

\infer[]
    {P \vdash [\ ] }
    { } \ \ \ \ \ \ \ \
  \infer[\text{if}~( A \gets A_1, \ldots, A_n) \in P]
    {P \vdash \sigma A }
    { P \vdash \sigma A_1 \quad \cdots \quad P \vdash \sigma A_n } 
  \end{array}
\]}
\end{definition}

In contrast, the SLD-resolution rule could be presented in the following form: 
$$B_1, \ldots , B_j, \ldots , B_n \leadsto_P \sigma B_1, \ldots,
\sigma A_1, \ldots, \sigma A_n, \ldots , \sigma B_n $$ if $(A \gets
A_1, \ldots, A_n) \in P$, and $\sigma$ is the mgu of $A$ and $B_{j}$.
The derivation for $A$ succeeds when $A \leadsto_P [\ ]$; we use
$\leadsto_P^*$ to denote several steps of SLD-resolution.

At first sight, the difference between TM-resolution and
SLD-resolution may seem only to be notational.  Indeed, both $ListNat
\vdash \mathtt{list(cons(0,nil))}$ and\\ $ \mathtt{list(cons(0,nil))}
\leadsto^*_{ListNat} [\ ]$ by the above rules (see also
Figure~\ref{pic:tree}). However, $ListNat \nvdash
\mathtt{list(cons(x,y))}$ whereas $ \mathtt{list(cons(x,y))}
\leadsto^*_{ListNat} [\ ]$.  And, even more mysteriously, $GC \nvdash
\mathtt{connected(x,y)}$ while $\mathtt{connected(x,y)} \leadsto_{GC}
       [\ ]$.

In fact, TM-resolution reflects the \emph{theorem proving}
aspect of LP: the rules of Definition~\ref{def:resolution} can be used to
semi-decide whether a given term $t$ is entailed by $P$.  In contrast,
SLD-resolution reflects the \emph{problem solving} aspect of LP: using
the SLD-resolution rule, one asks whether, for a given $t$, a
substitution $\sigma$ can be found such that $P \vdash \sigma(t)$.
There is a subtle but important difference between these two aspects
of proof search.
 
For example, when considering the successful derivation $
\mathtt{list(cons(x,y))}$ $ \leadsto^*_{ListNat} [\ ]$, we assume that
$\mathtt{list(cons(x,y))}$ holds only relative to a computed
substitution, e.g.  $\mathtt{x \mapsto 0, \ y \mapsto nil}$.  Of
course this distinction is natural from the point of view of theorem
proving: $\mathtt{list(cons(x,y))}$ is not a ``theorem" in this
generality, but its special case, $\mathtt{list(cons(0,nil))}$,
is. Thus, $ListNat \vdash \mathtt{list(cons(0,nil))}$ but $ListNat
\nvdash \mathtt{list(cons(x,y))}$ (see also Figure~\ref{pic:tree}).
Similarly, $\mathtt{connected(x,y)} \leadsto_{GC} [\ ]$ should be read
as: $\mathtt{connected(x,y)}$ holds relative to the computed
substitution $\mathtt{y\mapsto x}$.

According to the soundness and completeness theorems for
SLD-resolution~\cite{Llo}, the derivation $\leadsto$ has
\emph{existential} meaning, i.e. when $\mathtt{list(cons(x,y))}
\leadsto^*_{ListNat} [\ ]$, the successful goal
$\mathtt{list(cons(x,y))}$ is not meant to be read as universally
quantified over $\mathtt{x}$ and $\mathtt{y}$.  In contrast,
TM-resolution proves a universal statement. So $GC \vdash
\mathtt{connected(x,x)}$ reads as: $\mathtt{connected(x,x)}$ is
entailed by GC for any $\mathtt{x}$.

 Much of our recent work has been devoted to formal understanding of
 the relation between the theorem proving and problem solving aspects
 of LP~\cite{JohannKK15,FK15}.  The type-theoretic semantics of
 TM-resolution, given by ``Horn clauses as types, $\lambda$-terms as
 proofs" is given in~\cite{FK15,FKSP16}.

Definition~\ref{def:resolution} gives rise to derivation
trees. E.g. the derivation (or, equivalently, the proof) for $ListNat
\vdash \mathtt{list(cons(0,nil))}$ can be represented by the following
derivation tree:

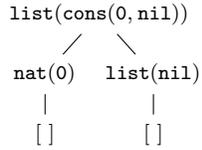
\begin{tikzpicture}[level 1/.style={sibling distance=18mm},
level 2/.style={sibling distance=18mm}, level 3/.style={sibling
distance=18mm},scale=.8,font=\footnotesize,baseline=(current bounding
box.north),grow=down,level distance=10mm]
\node (root) {$\mathtt{list(cons(0,nil))}$}
	child { node {$\mathtt{nat(0)}$}
	child { node {$[\ ]$}
	}}
  	child { node {$\mathtt{list(nil)}$}
		child { node {$[\ ]$}}};
  \end{tikzpicture} 
	
	In general, given a term $t$ and a program $P$, more than one
        derivation for $P \vdash t$ is possible.  For example, if we
        add a fifth clause to the program
        $ListNat$:\\ $5. \ \mathtt{list(cons(0,x)) \gets
          list(x)}$\\ then yet another, alternative, proof is possible
        for the extended program: $ListNat^+ \vdash
        \mathtt{list(cons(0,nil))}$ via Clause $5$:
	
	
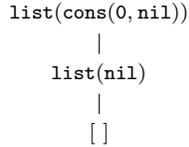
\begin{tikzpicture}[level 1/.style={sibling distance=18mm},
level 2/.style={sibling distance=18mm}, level 3/.style={sibling
distance=18mm},scale=.8,font=\footnotesize,baseline=(current bounding
box.north),grow=down,level distance=10mm]
\node (root) {$\mathtt{list(cons(0,nil))}$}
	  	child { node {$\mathtt{list(nil)}$}
		child { node {$[\ ]$}}};
  \end{tikzpicture}
	
	To reflect the choice of derivation strategies at every stage
        of the derivation, we introduce a new kind of node called an
        \emph{or-node}.  In our example, this would give us the tree
        shown in Figure~\ref{pic:tree}: note the $\bullet$-nodes.
	
	\begin{figure}[t]
\begin{center}
		\begin{tikzpicture}[level 1/.style={sibling distance=30mm},
level 2/.style={sibling distance=20mm}, level 3/.style={sibling
distance=18mm},scale=.8,font=\footnotesize,baseline=(current bounding
box.north),grow=down,level distance=8mm]
\node (root) {$\mathtt{list(cons(0,nil))}$}
  child {[fill] circle (2pt)
	child { node {$\mathtt{nat(0)}$}
	  child {[fill] circle (2pt)
	child { node {$[\ ]$}
	}}}
  	child { node {$\mathtt{list(nil)}$}
		  child {[fill] circle (2pt)
		child { node {$[\ ]$}}}}}
		  child {[fill] circle (2pt)
				child { node {$\mathtt{list(nil)}$}
				  child {[fill] circle (2pt)
		child { node {$[\ ]$}}}}};
  \end{tikzpicture} \ \ \ \ \
		\begin{tikzpicture}[level 1/.style={sibling distance=30mm},
level 2/.style={sibling distance=20mm}, level 3/.style={sibling
distance=18mm},scale=.8,font=\footnotesize,baseline=(current bounding
box.north),grow=down,level distance=8mm]
\node (root) {$\mathtt{list(cons(x,y))}$}
  child {[fill] circle (2pt)
	child { node {$\mathtt{nat(x)}$}}
  	child { node {$\mathtt{list(y)}$}}};
  \end{tikzpicture}
\end{center}
	\caption{\textbf{Left:} a coinductive tree for
          $\mathtt{list(cons(0,nil))}$ and the extended program
          $ListNat^+$. \textbf{Right:} a coinductive tree for
          $\mathtt{list(cons(x,y))}$ and $ListNat^+$. The
          $\bullet$-nodes mark different clauses applicable to every
          atom in the tree.}
\label{pic:tree} 
	\end{figure}
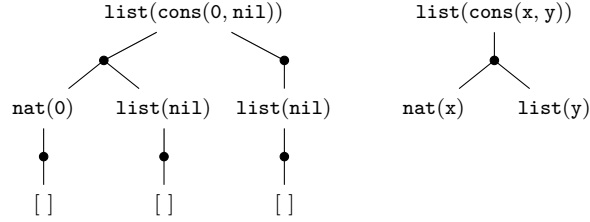
	
	This intuition is made precise in the following definition of
        a \emph{coinductive tree}, which first appeared
        in~\cite{KoP,KoPS} and was  refined in~\cite{JohannKK15}
        under the name of a rewriting tree. Note the use of mgms
        (rather than mgus) in the last item.

\begin{definition}[Coinductive tree]\label{def:cointree}
Let P be a logic program and $A$ be an atomic formula. The
\emph{coinductive tree} for $A$ is the possibly infinite tree T
satisfying the following properties.
\begin{itemize}
\item $A$ is the root of $T$
\item Each node in $T$ is either an and-node or an or-node
\item Each or-node is given by $\bullet$
\item Each and-node is an atom
\item For every and-node $A'$ occurring in $T$, if there is a clause $C_i$ in P of 
  the form $B_i \gets B^i_1,\ldots ,B^i_{n_i}$ for some $n_i$, such
  that $A' = \theta B_i$ for the mgm
  $\theta$, then $A'$ has an or-node, and that or-node has children given by and-nodes $\theta(B^i_j), \ldots
  ,\theta(B^i_{k})$, where $B_j, \ldots , B_k \subseteq B_1, \ldots , B_{n_i}$ and $B_j, \ldots , B_k$ is the maximal such set for which $\theta(B^i_j), \ldots
  ,\theta(B^i_{k})$ are distinct.
\end{itemize}
\end{definition}



Coinductive trees provide a convenient model for proofs by TM-resolution.  

Let us make one final observation on TM-resolution.
Generally, given a program $P$ and an atom $t$, one can prove that 

\begin{center}
$ t \leadsto^*_P [\ ]$ with computed substitution $\sigma$ if and only if $P
  \vdash \sigma t$.
\end{center}

This simple fact may leave the impression that proofs (and correspondingly
coinductive trees) for TM-resolution are in some sense fragments of
reductions by SLD-resolution.  Compare, for example, the right-hand tree of
Figure~\ref{pic:tree} before substitution with the larger left-hand tree
obtained after the substitution.  In this case, we could emulate the
problem solving aspect of SLD-resolution by using coinductive trees
and allowing the application of substitutions within coinductive trees, as was
proposed in~\cite{KoP2,JohannKK15,FK15}.  That works
perfectly for programs such as ListNat, but not for existential
programs: although there is a one step SLD-derivation for $
\mathtt{connected(x,y)} \leadsto_{GC} [\ ]$ (with $\mathtt{y \mapsto
  x}$), the TM-resolution proof for $ \mathtt{connected(x,y)}$ diverges
and gives rise to the following infinite coinductive tree:

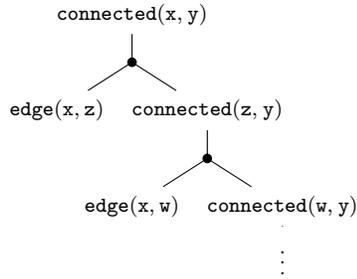
\vspace*{-0.1in}
	\begin{tikzpicture}[level 1/.style={sibling distance=30mm},
level 2/.style={sibling distance=25mm}, level 3/.style={sibling
distance=25mm},scale=.8,font=\footnotesize,baseline=(current bounding
box.north),grow=down,level distance=8mm]
  \node {$\mathtt{connected(x,y)}$}
   child {[fill] circle (2pt)
     child { node {$\mathtt{edge(x,z)}$}}
       child { node {$\mathtt{connected(z,y)}$}
        child {[fill] circle (2pt)
       child { node{$\mathtt{edge(x,w)}$}}
         child { node{$\mathtt{connected(w,y)}$}
				child {node{$\vdots$}}}}
       }}; 
  \end{tikzpicture}
	
Not only is the proof for $GC \vdash \mathtt{connected(x,y)}$ not a 
fragment of the derivation $\mathtt{connected(x,y)}
\leadsto_{GC} [\ ]$, but it also requires more  (infinitely many)
variables.  Thus, the operational semantics of TM-resolution and
SLD-resolution can be very different for existential programs, in
regard both to termination and to the number of variables involved.

This issue is largely orthogonal to that of non-termination. Consider the
non-terminating (but not existential) program Bad:

$\mathtt{bad(x)} \gets \mathtt{bad(x)}$\\ For Bad, the operational
behaviours of TM-resolution and SLD-resolution are similar: in both cases, derivations
do not terminate, and both require only finitely many variables.  Moreover,
such programs can be analysed using similar coinductive methods in TM-
and SLD-resolution~\cite{FKSP16,SimonBMG07}.

	The problems caused by existential variables are known in the
        literature on theorem proving and
        term-rewriting~\cite{Terese}.  In TRS~\cite{Terese},
        existential variables are not allowed to appear in rewriting
        rules, and in type inference based on term rewriting or
        TM-resolution, the restriction to non-existential programs is
        common~\cite{Jones97}.


So theorem-proving, in contrast to problem-solving, is modelled by
term-matching; term-matching gives rise to coinductive trees; and as
explained in the introduction and, in more detail, later, coinductive
trees give rise to laxness. So in this paper, we use laxness to model
coinductive trees, and thereby theorem-proving in LP, and we relate our
semantics with Bonchi and Zanasi's saturated semantics, which we believe 
primarily models the problem-solving aspect of logic programming.

Categorical semantics for existential programs, which are known to be
challenging for theorem proving, is a central contribution of
Section~\ref{sec:derivation} and of this paper.

 \section{Semantics for variable-free logic 
programs}\label{sec:parallel}

 In this section, we recall and develop the work of~\cite{KMP}, in regard to variable-free logic programs, i.e., we take $Var = \emptyset$ in Definition~\ref{df:syntax}.
 Variable-free logic programs are operationally equivalent to
 propositional logic programs, as substitutions 
play no role in derivations. In this (propositional) setting,
 coinductive trees coincide with the and-or derivation trees known in
 the LP literature~\cite{GC}, and this semantics appears as the ground case of
both lax semantics~\cite{KoPS} and saturated semantics~\cite{BonchiZ15}.

 \begin{proposition}\label{const:coal}
   For any set $\At$, there is a bijection between the set of
   variable-free logic programs over the set of atoms $\At$ and the set
   of $P_fP_f$-coalgebra structures on $\At$, where $P_f$ is the finite
   powerset functor on $Set$.
 \end{proposition}

 \begin{theorem}\label{constr:Gcoalg}
   Let $C(P_fP_f)$ denote the cofree comonad on $P_fP_f$. Then, given a logic
program $P$ over $\At$, equivalently $p:
   \At \longrightarrow P_f P_f(\At)$, the corresponding
   $C(P_fP_f)$-coalgebra $\overline{p}: \At \longrightarrow C(P_fP_f)(\At)$
sends an atom $A$ to the coinductive tree for $A$.
\end{theorem}

\begin{proof}

Applying the work of~\cite{W} to this setting, the cofree comonad is
in general determined as follows: $C(P_fP_f)(\At)$ is the limit of the
diagram
 $$\ldots \longrightarrow \At \times P_fP_f(\At \times P_fP_f(\At))
\longrightarrow \At \times P_fP_f(\At) \longrightarrow \At$$ with
maps determined by the projection $\pi_0:At\times
P_fP_f(At)\longrightarrow At$, with applications of the functor $At
\times P_fP_f(-)$ to it.

 Putting $\At_0 = \At$ and $\At_{n+1} = \At \times
 P_fP_f\At_n$, and defining the cone
 \begin{eqnarray*}
   p_0 & = & id: \At \longrightarrow \At ( = \At_0)\\
   p_{n+1} & = & \langle id, P_fP_f(p_n) \circ p \rangle : \At
   \longrightarrow \At \times P_fP_f \At_n ( = \At_{n+1})
 \end{eqnarray*}
 the limiting property of the diagram determines the coalgebra
 $\overline{p}: \At \longrightarrow C(P_fP_f)(\At)$.  The image
 $\overline{p}(A)$ of an atom $A$ is given by an element of the limit,
 equivalently a map from $1$ into the limit, equivalently a cone of
 the diagram over $1$.

To give the latter is equivalent to giving an element $A_0$ of $At$,
specifically $p_0(A) = A$, together with an element $A_1$ of $At\times
P_fP_f(At)$, specifically $p_1(A) = (A,p_0(A)) = (A,p(A))$, together
with an element $A_2$ of $At\times P_fP_f(At\times P_fP_f(At))$,
etcetera. The definition of the coinductive tree for $A$ is inherently
coinductive, matching the definition of the limit, and with the first
step agreeing with the definition of $p$. Thus it follows by
coinduction that $\overline{p}(A)$ can be identified with the
coinductive tree for $A$.
\end{proof}

\begin{example} \label{ex:free}
Let $At$ consist of atoms $\mathtt{A,B,C}$ and $\mathtt{D}$. Let $P$
denote the logic program
 \begin{eqnarray*}
 \mathtt{A} & \gets & \mathtt{B,C} \\
\mathtt{A} & \gets & \mathtt{B,D} \\
\mathtt{D} & \gets & \mathtt{A,C}\\
 \end{eqnarray*}
So $p(\mathtt{A}) = \{ \{ \mathtt{B,C}\} , \{ \mathtt{B,D} \} \}$,
$p(\mathtt{B}) = p(\mathtt{C}) = \emptyset$, and $p(\mathtt{D}) = \{
\{ \mathtt{A,C}\} \}$.

Then $p_0(\mathtt{A}) = \mathtt{A}$, which is the root of the
coinductive tree for $\mathtt{A}$.

Then $p_1(\mathtt{A}) = (\mathtt{A},p(\mathtt{A})) = (\mathtt{A},\{ \{
\mathtt{B,C}\} , \{ \mathtt{B,D} \} \})$, which consists of the same
information as in the first three levels of the coinductive tree for
$\mathtt{A}$, i.e., the root $\mathtt{A}$, two or-nodes, and below
each of the two or-nodes, nodes given by each atom in each antecedent
of each clause with head $\mathtt{A}$ in the logic program $P$: nodes
marked $\mathtt{B}$ and $\mathtt{C}$ lie below the first or-node, and
nodes marked $\mathtt{B}$ and $\mathtt{D}$ lie below the second
or-node, exactly as $p_1(\mathtt{A})$ describes.

Continuing, note that $p_1(\mathtt{D}) = (\mathtt{D},p(\mathtt{D})) =
(\mathtt{D},\{ \{\mathtt{A,C}\} \})$. So
\[
\begin{array}{ccl}
p_2(\mathtt{A}) & = & (\mathtt{A},P_fP_f(p_1)(p(\mathtt{A})))\\
            & = & (\mathtt{A},P_fP_f(p_1)( \{ \{ \mathtt{B,C}\} , \{ \mathtt{B,D} \} \}))\\
            & = & (\mathtt{A}, \{ \{( \mathtt{B},\emptyset),(\mathtt{C},\emptyset)\} , \{( \mathtt{B},\emptyset),(\mathtt{D},\{ \{\mathtt{A,C}\} \}) \} \})
\end{array}
\]
which is the same information as that in the first five levels of the
coinductive tree for $\mathtt{A}$: $p_1(\mathtt{A})$ provides the
first three levels of $p_2(\mathtt{A})$ because $p_2(\mathtt{A})$ must
map to $p_1(\mathtt{A})$ in the cone; in the coinductive tree, there
are two and-nodes at level 3, labelled by $\mathtt{A}$ and
$\mathtt{C}$. As there are no clauses with head $\mathtt{B}$ or
$\mathtt{C}$, no or-nodes lie below the first three of the and-nodes
at level 3. However, there is one or-node lying below $\mathtt{D}$, it
branches into and-nodes labelled by $\mathtt{A}$ and $\mathtt{C}$,
which is exactly as $p_2(\mathtt{A})$ tells us. For picture of this tree, see~Figure~\ref{fig:gtree}.
\end{example}

\begin{figure}[!h]
\begin{center}
  \begin{tikzpicture}[scale=0.9,baseline=(current bounding box.north),grow=down,level distance=7mm,  level 1/.style={sibling
distance=25mm}, level 2/.style={sibling
distance=10mm}, sibling distance = 10mm]
\node {$\mathtt{A}$}
child { [fill] circle (2pt)
  child { node {$\mathtt{B}$}}
	 child { node {$\mathtt{C}$}}}
     child {[fill] circle (2pt)
         child {node {$\mathtt{B}$}}
				child {node {$\mathtt{D}$}
            child {[fill] circle (2pt)
               child { node {$A$}
							child { node {$\ldots$} }}
							child { node {$C$} }}}}; 
  \end{tikzpicture}
\end{center}
\caption{The coinductive tree for $\texttt{A}$ and the program $P$ from Example~\ref{ex:free}.}
\label{fig:gtree} 
\end{figure}
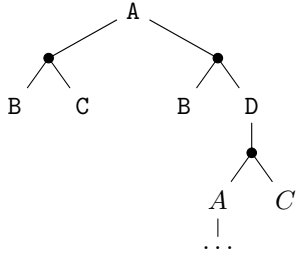

 \section{Lax semantics for logic programs}\label{sec:recall}

We now lift the restriction on $Var = \emptyset$ in
Definition~\ref{df:syntax} and consider first-order terms and atoms
in full generality.

%
%
There are several equivalent ways in which to describe the Lawvere theory generated by a signature. So, for precision, in this
paper, we define the \emph{Lawvere theory} $\mathcal{L}_{\Sigma}$ \emph{generated by} a
signature $\Sigma$ as follows:
$\texttt{ob}(\mathcal{L}_{\Sigma})$ is the set of natural numbers.
For each natural number $n$, let $x_1,\ldots ,x_n$ be a specified list
of distinct variables. Define $\mathcal{L}_{\Sigma}(n,m)$ to be the
set of $m$-tuples $(t_1,\ldots ,t_m)$ of terms generated by the
function symbols in $\Sigma$ and variables $x_1,\ldots ,x_n$. Define
composition in $\mathcal{L}_{\Sigma}$ by substitution.

 One can readily check that these constructions satisfy the axioms for
 a category, with $\mathcal{L}_{\Sigma}$ having
 strictly associative finite products given by the sum of natural
 numbers. The terminal object of $\mathcal{L}_{\Sigma}$ is the natural
 number $0$. There is a canonical identity-on-objects functor from $Nat^{op}$ to $\ls$,
just as there is for any Lawvere theory, and it strictly preserves finite products.

 \begin{example}\label{ex:arrows}
   Consider ListNat. 
   The constants $\mathtt{O}$ and $\mathtt{nil}$ are maps
   from $0$ to $1$ in $\mathcal{L}_{\Sigma}$, $\mathtt{s}$ is modelled by
   a map from $1$ to $1$, and $\mathtt{cons}$ is modelled by a map from
   $2$ to $1$. The term $\mathtt{s(0)}$ is the map
   from $0$ to $1$ given by the composite of the maps modelling
   $\mathtt{s}$ and $\mathtt{0}$. 
 \end{example}

Given an arbitrary logic program $P$ with signature $\Sigma$,
 we can extend the set $At$ of atoms for a variable-free
 logic program to the functor $At:\ls^{op}
 \rightarrow Set$ that sends a natural number $n$ to the set of all
 atomic formulae, with variables among $x_1,\ldots ,x_n$, generated by
 the function symbols in $\Sigma$ and by the predicate symbols in $P$. 
A map $f:n \rightarrow m$ in $\ls$ is sent to
 the function $At(f):At(m) \rightarrow At(n)$ that sends an atomic
 formula $A(x_1, \ldots,x_m)$ to $A(f_1(x_1, \ldots ,x_n)/x_1, \ldots
 ,f_m(x_1, \ldots ,x_n)/x_m)$, i.e., $At(f)$ is defined by
 substitution.

 As explained in the Introduction and in~\cite{KMP}, we cannot model a
 logic program by a natural transformation of the form
 $p:At\longrightarrow P_fP_fAt$ as naturality breaks down, e.g., in
 ListNat. So, in~\cite{KoP,KoPS}, we relaxed naturality to lax
 naturality.  In order to define it, we extended
 $At:\ls^{op}\rightarrow Set$ to have codomain $Poset$ by composing
 $At$ with the inclusion of $Set$ into $Poset$. Mildly overloading
 notation, we denote the composite by $At:\ls^{op}\rightarrow Poset$.
 
 \begin{definition}
   Given functors $H,K:\ls^{op} \longrightarrow Poset$, a {\em lax
     transformation} from $H$ to $K$ is the assignment to each object
   $n$ of $\ls$, of an order-preserving function $\alpha_n: Hn
   \longrightarrow Kn$ such that for each map $f:n \longrightarrow m$
   in $\ls$, one has $(Kf)(\alpha_m) \leq (\alpha_{n})(Hf)$, pictured
   as follows:
\begin{diagram}
Hm & \rTo^{\alpha_m} & Km \\
\dTo<{Hf} & \geq & \dTo>{Kf} \\
Hn & \rTo_{\alpha_n} & Kn
\end{diagram}
 \end{definition}
 
 Functors and lax transformations, with pointwise composition, form a
 locally ordered category denoted by $Lax(\ls^{op},Poset)$. Such
 categories and generalisations have been studied extensively, e.g.,
 in~\cite{Ben,BKP,Kelly,KP1}.

 \begin{definition}\label{def:poset}
 Define $P_f:Poset\longrightarrow Poset$ by letting
   $P_f(P)$ be the partial order given by the set of finite subsets of
   $P$, with $A\leq B$ if for all $a \in A$, there exists
   $b \in B$ for which $a\leq b$ in $P$, with behaviour on maps
   given by image. Define $P_c$ similarly but with countability
   replacing finiteness.
 \end{definition}

We are not interested in arbitrary posets in modelling logic
programming, only those that arise, albeit inductively, by taking
subsets of a set qua discrete poset. So we gloss over the fact that,
for an arbitrary poset $P$, Definition~\ref{def:poset} may yield
factoring, with the underlying set of $P_f(P)$ being a quotient of the
set of subsets of $P$. It does not affect the line of development
here.

\begin{example}\label{ex:listnat2}
Modelling Example~\ref{ex:listnat}, ListNat generates a lax
transformation of the form $p:At\longrightarrow P_fP_fAt$ as follows:
$At(n)$ is the set of atomic formulae in $ListNat$ with at most $n$
variables.

For example, $At(0)$ consists of $\mathtt{nat(0)}$,
$\mathtt{nat(nil)}$, $\mathtt{list(0)}$, $\mathtt{list(nil)}$,
$\mathtt{nat(s(0))}$, $\mathtt{nat(s(nil))}$, $\mathtt{list(s(0))}$,
$\mathtt{list(s(nil))}$, $\mathtt{nat(cons(0, 0))}$,
$\mathtt{nat(cons( 0, nil))}$, \\ $\mathtt{nat(cons (nil, 0))}$,
$\mathtt{nat(cons( nil, nil))}$, etcetera.

Similarly, $At(1)$ includes all atomic formulae containing at most one
(specified) variable $x$, thus all the elements of $At(0)$ together
with $\mathtt{nat(x)}$, $\mathtt{list(x)}$, $\mathtt{nat(s(x))}$,
$\mathtt{list(s(x))}$, $\mathtt{nat(cons( 0, x))}$, $\mathtt{nat(cons
  (x, 0))}$, $\mathtt{nat(cons (x, x))}$, etcetera.

The function $p_n:At(n)\longrightarrow P_fP_fAt(n)$ sends each element
of $At(n)$, i.e., each atom $A(x_1,\ldots ,x_n)$ with variables among
$x_1,\ldots ,x_n$, to the set of sets of atoms in the antecedent of
each unifying substituted instance of a clause in $P$ with head for
which a unifying substitution agrees with $A(x_1,\ldots ,x_n)$.

Taking $n=0$, $\mathtt{nat(0)}\in At(0)$ is the head of one clause,
and there is no other clause for which a unifying substitution will
make its head agree with $\mathtt{nat(0)}$. The clause with head
$\mathtt{nat(0)}$ has the empty set of atoms as its tail, so
$p_0(\mathtt{nat(0)}) = \{ \emptyset \}$.

Taking $n=1$, $\mathtt{list(cons( x, 0))}\in At(1)$ is the head of one
clause given by a unifying substititution applied to the final clause
of ListNat, and accordingly $p_1(\mathtt{list(cons (x, 0))}) = \{ \{
\mathtt{nat(x)},\mathtt{list(0)} \} \}$.

The family of functions $p_n$ satisfy the inequality required to form
a lax transformation precisely because of the allowability of
substitution instances of clauses, as in turn is required to model
logic programming. The family does not satisfy the strict requirement
of naturality as explained in the introduction.
\end{example}

\begin{example}\label{ex:lp2}
Attempting to model Example~\ref{ex:lp}, that of graph connectedness, GC, by mimicking the modelling of
ListNat 
in Example~\ref{ex:listnat2}, i.e., defining the function \mbox{$p_n:At(n)\longrightarrow P_fP_fAt(n)$} by sending each element
of $At(n)$, i.e., each atom $A(x_1,\ldots ,x_n)$ with variables among
$x_1,\ldots ,x_n$, to the set of sets of atoms in the antecedent of
each unifying substituted instance of a clause in $P$ with head for
which a unifying substitution agrees with $A(x_1,\ldots ,x_n)$, fails.

Consider the clause
\[
 \mathtt{connected(x,y)}  \gets  \mathtt{edge(x,z)},
 \mathtt{connected(z,y)} 
\]

Modulo possible renaming of variables, the head of the clause, i.e.,
the atom $\mathtt{connected(x,y)}$, lies in $At(2)$ as it has two
variables. There is trivially only one substituted instance of a clause in GC with head for
which a unifying substitution agrees with $\mathtt{connected(x,y)}$, and the singleton set 
consisting of the set of atoms in its antecedent
is $\{ \{ \mathtt{edge(x,z)},\mathtt{connected(z,y)} \} \}$, which does not lie in $P_fP_fAt(2)$ as it has three variables appear in it rather than two. See Section~\ref{sec:backr}
 for a picture of the coinductive tree for $\mathtt{connected(x,y)}$.

We dealt with that inelegantly in~\cite{KoP}: in order to force
$p_2(\mathtt{connected(x,y)})$ to lie in $P_fP_fAt(2)$ and model GC in any reasonable 
sense, we
allowed substitutions for $z$ in $\{ \{ \mathtt{edge(x,z)},\mathtt{connected(z,y)} \} \}$
by any term on $x,y$ on the basis that
there is no unifying such, so we had better allow all
possibilities. So, rather than modelling the clause directly,
recalling that $At(2)\subseteq At(3)\subseteq At(4)$, etcetera, modulo
renaming of variables, we put

{\small{
\begin{eqnarray*}
p_2(\mathtt{connected(x,y)}) & = & \{
\{\mathtt{edge(x,x)},\mathtt{connected(x,y)}\},
\{\mathtt{edge(x,y)},\mathtt{connected(y,y)}\}\}\\ p_3(\mathtt{connected(x,y)})
& = & \{ \{\mathtt{edge(x,x)},\mathtt{connected(x,y)}\},
\{\mathtt{edge(x,y)},\mathtt{connected(y,y)}\},\\ & &
\{\mathtt{edge(x,z)},\mathtt{connected(z,y)}\}
\}\\ p_4(\mathtt{connected(x,y)}) & = & \{
\{\mathtt{edge(x,x)},\mathtt{connected(x,y)}\},
\{\mathtt{edge(x,y)},\mathtt{connected(y,y)}\},\\ & &
\{\mathtt{edge(x,z)},\mathtt{connected(z,y)}\},\{\mathtt{edge(x,w)},\mathtt{connected(w,y)}\}
\}
\end{eqnarray*}}}
etcetera: for $p_2$, as only two variables $x$ and $y$ appear in any
element of $P_fP_fAt(2)$, we allowed substitution by either $x$ or $y$
for $z$; for $p_3$, a third variable may appear in an element of
$P_fP_fAt(3)$, allowing an additional possible subsitution; for $p_4$,
a fourth variable may appear, etcetera.

Countability arises if a unary symbol $s$ is added to GC, as in that
case, for $p_2$, not only did we allow $x$ and $y$ to be substituted
for $z$, but we also allowed $s^n(x)$ and $s^n(y)$ for any $n>0$, and
to do that, we replaced $P_fP_f$ by $P_cP_f$, allowing for the
countably many possible substitutions.

Those were  inelegant
decisions, but they allowed us to give some kind of model of all logic programs.
We shall revisit this in Section~\ref{sec:derivation}.
\end{example}


We shall refine lax semantics to account for existential variables later, so
for the present, we shall ignore Example~\ref{ex:lp2} and only analyse
semantics for logic programs without existential variables such as in Example~\ref{ex:listnat2}.
Specifically, we shall analyse the relationship between a lax transformation
$p:At\longrightarrow P_fP_fAt$ and
$\overline{p}:At\longrightarrow C(P_fP_f)At$, the corresponding
coalgebra for the cofree comonad $C(P_fP_f)$ on $P_fP_f$.

We recall the central abstract result of~\cite{KoP}, the notion of an
``oplax" map of coalgebras being required to match that of lax
transformation. Notation of the form \mbox{$H$-$coalg$} refers to
coalgebras for an endofunctor $H$, while notation of the form
\mbox{$C$-$Coalg$} refers to coalgebras for a comonad $C$. The subscript
$oplax$ refers to oplax maps and, given an
endofunctor $E$ on $Poset$, the notation $Lax(\ls^{op},E)$ denotes the
endofunctor on $Lax(\ls^{op},Poset)$ given by post-composition with
$E$; similarly for a comonad.

 \begin{theorem}\label{main}\cite{KoP}
 For any locally ordered endofunctor $E$ on $Poset$, if $C(E)$ is
 the cofree comonad on $E$, then there is a canonical isomorphism
 \[
 Lax(\ls^{op},E)\mbox{-}coalg_{oplax} \simeq
 Lax(\ls^{op},C(E))\mbox{-}Coalg_{oplax}
 \]
 \end{theorem}

Theorem~\ref{main} tells us that for any endofunctor $E$ on $Poset$, the relationship
between $E$-coalgebras and $C(E)$-coalgebras extends pointwise from $Poset$ 
to $Lax(\ls^{op},Poset)$ providing one matches lax natural transformations by oplax
maps of coalgebras. It follows that, given an endofunctor $E$ on $Poset$ with cofree 
comonad $C(E)$, the cofree comonad for the endofunctor on $Lax(\ls^{op},Poset)$ sending
$H:\ls^{op}\longrightarrow Poset$ to the composite $EH:\ls^{op}\longrightarrow Poset$
sends $H$ to the composite $C(E)H$. Taking the example $E = P_fP_f$ allows us to conclude the following.

\begin{corollary}\label{oldcor}\cite{KoP}
$Lax(\ls^{op},C(P_fP_f))$ is the cofree
comonad on $Lax(\ls^{op},P_fP_f)$.
\end{corollary}
Corollary~\ref{oldcor} means that there is a natural bijection between lax transformations
\[
p:At\longrightarrow P_fP_fAt
\]
and lax transformations 
\[
\overline{p}:At\longrightarrow C(P_fP_f)At
\]
subject to the two conditions required of a coalgebra of a
comonad given pointwise, thus by applying the construction of Theorem~\ref{constr:Gcoalg} pointwise. So it is the abstract result we need in order to characterise the coinductive trees generated by logic programs with no existential variables, extending
Theorem~\ref{constr:Gcoalg} as follows.

\begin{theorem}\label{constr:Gcoalgcount}
Let $C(P_fP_f)$ denote the cofree comonad on the endofunctor $P_fP_f$ on $Poset$. Then, given a logic program $P$ with no existential variables on $At$, defining $p_n(A(x_1,\ldots ,x_n))$ to be the set of sets of atoms in each antecedent of each unifying substituted instance of a clause in $P$ with head for which a unifying substitution agrees with $A(x_1,\ldots ,x_n)$, the corresponding
$Lax(\ls^{op},C(P_fP_f))$-coalgebra $\overline{p}:At\longrightarrow C(P_fP_f)At$ sends an atom
$A(x_1,\ldots ,x_n)$ to the coinductive tree for $A(x_1,\ldots ,x_n)$.
\end{theorem}

\begin{proof}
The absence of existential variables ensures that any variable that appears in the antecedent of a clause must also appear in its head. So every atom in every antecedent
of every unifying substituted instance of a clause in $P$ with head for which a unifying
substitution agrees with $A(x_1,\ldots ,x_n)$ actually lies in $At(n)$. Moreover, there
are only finitely many sets of sets of such atoms. So the construction of each $p_n$ is well-defined, i.e., the image of $A(x_1,\ldots ,x_n)$  lies in $P_fP_fAt(n)$. 
The $p_n$'s collectively form a lax transformation from $At$ to $P_fP_fAt$ as substitution preserves the truth of a clause. 

By Corollary~\ref{oldcor}, $\overline{p}$ is determined pointwise. So, to construct it, we may fix $n$ and follow the proof of Theorem~\ref{constr:Gcoalg}, consistently replacing $At$ by $At(n)$. To complete the proof, observe that the construction of $p$ from a logic program $P$ matches the construction of the coinductive tree for an atom $A(x_1,\ldots ,x_n)$ if $P$ has no existential variables. So following the proof of Theorem~\ref{constr:Gcoalg} completes this proof.
\end{proof}

Theorem~\ref{constr:Gcoalgcount} models the coinductive trees generated by
ListNat as the latter has no existential
variables, but for GC, as explained in Example~\ref{ex:lp2}, the natural
construction of $p$ did
\emph{not} model the clause
\[
 \mathtt{connected(x,y)}  \gets  \mathtt{edge(x,z)},
 \mathtt{connected(z,y)} 
\]
directly, and so its extension \emph{a fortiori} could \emph{not} model the 
coinductive trees generated by $\mathtt{connected(x,y)}$.

For arbitrary logic programs, the way we defined $\overline{p}(A(x_1,\ldots ,x_n))$ in earlier papers such as~\cite{KoPS} was  in terms of a variant of the coinductive tree
generated by $A(x_1,\ldots ,x_n)$  in two key ways:
\begin{enumerate}
\item coinductive trees allow new variables to be introduced as one passes down the tree, e.g., with
\[
 \mathtt{connected(x,y)}  \gets  \mathtt{edge(x,z)},
 \mathtt{connected(z,y)}
\]
appearing directly in it, whereas, if we extended the construction of $p$ in Example~\ref{ex:lp2},
$\overline{p_1}(\mathtt{connected(x,y)})$ would not model such a clause
directly, but would rather substitute terms on $x$ and $y$ for $z$,
continuing inductively as one proceeds.
\item coinductive trees are finitely branching, as one expects in
  logic programming, whereas $\overline{p}(A(x_1,\ldots ,x_n))$ could
  be infinitely branching, e.g., for GC with an additional unary
  operation $s$.
\end{enumerate}

\section{Saturated semantics for logic programs}\label{sec:sat}

Bonchi and Zanasi's saturated semantics approach to modelling logic programming
in~\cite{BZ} was to consider $P_fP_f$ as we did in~\cite{KoP}, sending
$At$ to $P_fP_fAt$, but to ignore the inherent laxness, replacing
$Lax(\ls^{op},Poset)$ by $[ob(\ls),Set]$, where $ob(\ls)$ is the set
of objects of $\ls$ treated as a discrete category, i.e., as a category containing
only identity maps. Their central construction may be seen in a more
axiomatic setting as follows. 

For any small category $C$, let $ob(C)$ denote the discrete subcategory with the same
objects as $C$, with inclusion $I:ob(C)\longrightarrow C$. Then the functor
\[
[I,Set]:[C,Set]\longrightarrow [ob(C),Set]
\]
has a right adjoint given by right Kan extension, and that remains true when
one extends from $Set$ to any complete category, and it all enriches, e.g., over $Poset$~\cite{K}. As $ob(C)$ has no non-trivial arrows, the right Kan extension is a product, given by

\[
(ran_I H)(c) = \prod_{d \in C} Hd^{C(c,d)}
\]
By the Yoneda lemma, to give a natural transformation from $K$ to $(ran_I H)(-)$ is equivalent to giving a natural, or equivalently in this setting, a ``not
necessarily natural", transformation from $KI$ to $H$. Taking $C = \ls^{op}$ gives exactly Bonchi and Zanasi's formulation of saturated semantics~\cite{BZ}. 

It was the fact of the existence of the right adjoint, rather
than its characterisation as a right Kan extension, that enabled
Bonchi and Zanasi's constructions of saturation and desaturation,
but the description as a right Kan extension informed their syntactic
analysis.

Note for later that products in $Poset$ are given pointwise, so agree with products in $Set$. So if we replace $Set$ by $Poset$ here, and if $C$ is an ordinary category without any non-trivial $Poset$-enrichment, the right Kan extension would yield the same set as above, with an order on it determined by that on $H$.

In order to unify saturated semantics with lax semantics, we need to rephrase Bonchi and Zanasi's formulation a little. Upon close inspection, one can see that, in their semantics, they only used objects of $[ob(\ls)^{op},Set]$, equivalently $[ob(\ls),Set]$, of the form $HI$ for some $H:\ls^{op}\longrightarrow Set$~\cite{BZ}. That allows us, while making no substantive change to their body of work, to reformulate it a little, in axiomatic terms, as follows.

Let $[C,Set]_d$ denote the category of
functors from $C$ to $Set$ and ``not necessarily natural"
transformations between them, i.e., a map from $H$ to $K$ consists of,
for all $c \in C$, a function $\alpha_c:Hc\longrightarrow Kc$, without
demanding a naturality condition. The functor $[I,Set]:[C,Set]\longrightarrow [ob(C),Set]$ factors through the inclusion of $[C,Set]$ into $[C,Set]_d$ as follows:

\begin{diagram}
[C,Set] & \rTo & [C,Set]_d & \rTo & [ob(C),Set]
\end{diagram}
In this decomposition, the functor from $[C,Set]_d$ to $[ob(C),Set]$ sends a functor $H:C\longrightarrow Set$ to its restriction $HI$ to $ob(C)$ and is fully faithful. Because it is fully faithful, it follows that the inclusion of $[C,Set]$ into $[C,Set]_d$  has a right adjoint also given by right Kan extension. 

Thus one can reprhrase Bonchi and Zanasi's work to assert that the central mathematical fact that supports saturated semantics is
that the inclusion
\[
[\ls^{op},Set]\longrightarrow [\ls^{op},Set]_d
\]
has a right adjoint that sends a functor $H:\ls^{op}\longrightarrow Set$ to
the right Kan extension $ran_I HI$ of the composite $HI:ob(\ls)^{op}\longrightarrow Set$
along the inclusion $I:ob(\ls) \longrightarrow \ls$.

We can now unify lax semantics with saturated semantics by developing a precise body of theory that relates the inclusion
\[
J:[C,Set] \longrightarrow [C,Set]_d
\]
which has a right adjoint that sends $H:C\longrightarrow Set$ to $ran_I HI$, with the inclusion
\[
J:[C,Poset] \longrightarrow Lax(C,Poset)
\]
which also has a right adjoint, that right adjoint being given by a restriction of the
right Kan extension $ran_I HI$ of the composite $HI:ob(C)\longrightarrow Poset$
along the inclusion $I:ob(C)\longrightarrow C$. 

The existence of the right adjoint
follows from the main result of~\cite{BKP}, but we give an independent proof here
and a description of it in terms of right Kan extensions in order to show that Bonchi
and Zanasi's explicit constructions of saturation and desaturation apply equally
in this setting.

Consider the inclusions
\[
[C,Poset] \longrightarrow Lax(C,Poset) \longrightarrow [C,Poset]_d
\]
As we have seen, the composite has a right adjoint sending $H:C\longrightarrow Poset$ to $ran_I HI$. So, to describe a right adjoint to  
\mbox{$J:[C,Poset]\longrightarrow Lax(C,Poset)$}, we need to
restrict $ran_I HI$ so that to give a natural transformation from $K$
into the restriction $R(H)$ of $ran_I HI$ is equivalent to giving a map from $H$ to $K$ in $[C,Poset]_d$ that
satisfies the condition that, for all $f:c\longrightarrow d$, one has
$Hf.\alpha_c \leq \alpha_d.Kf$. This can be done by defining $R(H)$ to be an \emph{inserter},
which is a particularly useful kind of limit that applies to locally ordered categories and is
a particular kind of generalisation of the notion of equaliser. 

\begin{definition}\label{def:inserter}\cite{BKP}
Given parallel maps $f,g:X\longrightarrow Y$ in a locally ordered category $K$, 
an \emph{inserter} from $f$ to $g$ is an object $Ins(f,g)$ of $K$ together with a map
$i:Ins(f,g)\longrightarrow X$ such that $fi \leq gi$ and is universal such, i.e., for any
object $Z$ and map $z:Z\longrightarrow X$ for which $fz\leq gz$, there is a unique
map $k:Z\longrightarrow Ins(f,g)$ such that $ik = z$. Moreover, for any such $z$ and $z'$
for which $z\leq z'$, then $k\leq k'$, where $k$ and $k'$ are induced by $z$ and $z'$ respectively.
\end{definition}

An inserter is a form of limit. Taking $K$ to be $Poset$, the poset $Ins(f,g)$ is
given by the full sub-poset of $X$ determined by $\{ x\in X|f(x)\leq g(x)\}$. Being
limits, inserters in functor categories are determined pointwise. 

\begin{theorem}\label{thm:R}
The right adjoint $R$  to the 
inclusion $J:[C,Poset]\longrightarrow Lax(C,Poset)$ sends $H:C\longrightarrow Poset$
to the inserter in $[C,Poset]$ from $\delta_1$ to $\delta_2$
\[
\delta_1,\delta_2:(ran_I H)(-) = \prod_{d \in C} Hd^{C(-,d)}\longrightarrow
\prod_{d,d' \in C} Hd'^{C(-,d)\times C(d,d')}
\]
where $\delta_1$ and $\delta_2$ are defined to be equivalent, by Currying, to $(d,d')$-indexed collections of maps of the form
\[
(\delta_1)_{(d,d')},(\delta_2)_{(d,d')}:C(-,d)\times C(d,d')\times \prod_{d \in C} Hd^{C(-,d)}\longrightarrow
Hd'
\]

which, in turn, are defined as follows:
\begin{enumerate}
\item the $(d,d')$-component of $\delta_1c$ is determined by composing
\[
\circ_C\times id:C(c,d)\times C(d,d') \times \prod_{d \in C} Hd^{C(c,d)} \longrightarrow
C(c,d') \times \prod_{d \in C} Hd^{C(c,d)}
\]
with the evaluation of the product at $d'$
\item
the $(d,d')$-component of $\delta_2c$ is determined by evaluating the product at $d$
\[
C(c,d)\times C(d,d') \times \prod_{d \in C} Hd^{C(c,d)} \longrightarrow
C(d,d') \times Hd
\]
then composing with
\begin{diagram}
C(d,d') \times Hd & \rTo^{H\times id} &
Hd'^{Hd} \times Hd & \rTo^{eval} & 
Hd'
\end{diagram}
\end{enumerate}
\end{theorem}

Although the statement of the theorem is complex, the proof is routine. One simply needs
to check that $\delta_1$ and $\delta_2$ are natural, which they routinely are, and that the inserter satisfies the universal property we seek, which it does by construction.

Bonchi and Zanasi's saturation and desaturation
constructions remain exactly the same: the saturation of $p:At\longrightarrow P_fP_fAt$
is a natural transformation $\overline{p}:At\longrightarrow ran_I P_fP_fAtI$ that 
factors through $Ins(\delta_1,\delta_2)$ without any change whatsoever to its construction, 
that being so because of the fact of $p$ being lax.

With this result in hand, it is routine to work systematically through Bonchi
and Zanasi's papers, using their saturation and desaturation constructions exactly
as they had them, without discarding the inherent laxness that logic
programming, cf data refinement, possesses.  

So this unifies lax semantics, which flows
from, and may be seen as an instance of, Tony Hoare's semantics for data refinement~\cite{HH,HH1,KP1}, with saturated semantics and its more denotational flavour~\cite{BM}.

 \section{Lax semantics for logic progams refined: existential variables}\label{sec:derivation}

In Section~\ref{sec:recall}, following~\cite{KoP}, we gave lax semantics for logic programs without existential variables, such as ListNat. In particular, we modelled the coinductive trees
they generate.  Restriction to non-existential examples such as ListNat is common for implementational reasons~\cite{KoPS,JohannKK15,FK15,FKSP16}, so Section~\ref{sec:recall}
allowed the modelling of coinductive trees for a natural class of logic
programs. 

Nevertheless, we would like to model coinductive trees generated by logic programming in full generality, including
examples such as that of GC. We need to refine the lax semantics of Section~\ref{sec:recall}
in order to do so, and, having just unified lax semantics with saturated semantics in Section~\ref{sec:sat}, we would like to retain that unity in making such a refinement. So that is what we do in this section.

We initially proposed such a refinement in the workshop paper~\cite{KoP1} that this paper extends, but since the workshop, we have found a further refinement that strengthens the relationship with the modelling of local state~\cite{PP2}. So our constructions here are a little different to those in~\cite{KoP1}.

In order to model coinductive trees, it follows from
Example~\ref{ex:lp2} that the endofunctor $Lax(\ls^{op},P_fP_f)$ on
$Lax(\ls^{op},Poset)$ that sends $At$ to $P_fP_fAt$, needs to be
refined as $\{ \{\mathtt{edge(x,z),connected(z,y)}\}\}$ is not an
element of $P_fP_fAt(2)$ as it involves three variables $x$, $y$ and
$z$. In general, we need to allow the image of $p_n$ to lie in the set given by applying
$P_fP_f$ to a superset of $At(n)$, one that includes $At(m)$ for all
$m\geq n$. However, we do not want to double-count: there are six injections 
of $2$ into $3$, inducing six
inclusions $At(2)\subseteq At(3)$, and one only wants to count each
atom in  $At(2)$ once. So we refine $P_fP_fAt(n)$ to become $P_fP_f(\int\! At(n))$, where $\int\! At$ is defined as follows.

Letting $Inj$ denote the category of natural numbers and injections, for any Lawvere theory $L$, there is a canonical identity-on-objects
functor $J:Inj^{op}\longrightarrow L$.  We define $\int\! At(n)$ to be the colimit of the composite functor
\begin{diagram}
n/Inj & \rTo^{cod} & Inj & \rTo^{J} & \ls^{op} & \rTo^{At} & Poset
\end{diagram}
This functor sends an injection $j:n\longrightarrow m$ to $At(m)$, with the $j$-th component of the colimiting cocone being of the
form $\rho_j:At(m)\longrightarrow \int\! At(n)$. The colimiting property is precisely the condition 
required to ensure no double-counting (see~\cite{Mac} or, for the enriched version,~\cite{K} of this construction in a general setting).

It is not routine to extend the construction of $\int\! At(n)$ to be functorial in $Inj$. So we mimic the construction on arrows used to define the monad for local state in~\cite{PP2}. 
We first used this idea in~\cite{KoP1} and we refine our use of it in this paper to make for a closer technical relationship with the semantics of local state in~\cite{PP2}: we do not fully 
understand the relationship yet, but there seems considerable potential based on the work here to make precise comparison between the role of variables in logic programming with that of worlds in modelling local state. 

In detail, the definition of $\int\! At(n)$ extends canonically to become a functor
\mbox{$\int\! At:\ls^{op}\longrightarrow Poset$} that sends a map
$f:n\longrightarrow n'$ in $\ls$ to the order-preserving function
\[
\int\! At(f):\int\! At(n') \longrightarrow \int\! At(n)
\]
determined by the colimiting property of $\int\! At(n')$ as follows: each $j'\in n'/Inj$ is, up to coherent isomorphism, the canonical injection $j':n'\longrightarrow n'+k$ for a unique natural
number $k$; that induces a cocone
\begin{diagram}
At(n'+k) & \rTo^{At(f+k)} & At(n+k) & \rTo^{\rho_j} & \int\! At(n)
\end{diagram}
where $j:n\longrightarrow n+k$ is the canonical injection of $n$ into $n+k$. It is routine
to check that this assignment respects compostion and identities, thus is functorial. 

There is nothing specific about $At$ in the above construction. So it generalises without fuss
from $At$ to apply to an arbitrary functor $H:\ls^{op}\longrightarrow Poset$.

In order to make the construction $\int\! H$ functorial in $H$, i.e., in order to make it respect 
maps $\alpha:H\Rightarrow K$, we need to refine 
$Lax(\ls^{op},Poset)$. Specifically, we need to restrict its maps to allow only those lax transformations $\alpha:H\Rightarrow K$ that are strict with respect to maps in $Inj$, i.e.,  those $\alpha$ such that  for any injection $i:n\longrightarrow m$, the diagram
\begin{diagram}
Hn & \rTo^{\alpha_n} & Kn \\
\dTo<{Hi} & & \dTo>{Ki} \\
Hm & \rTo_{\alpha_m} & Km 
\end{diagram}
commutes.
The reason for the restriction is that the colimit that defines $\int\! H(n)$ strictly
respects injections, so we need a matching condition on $\alpha$ in order to be able to define $\int\! \alpha (n)$.

Summarising this discussion yields the following:

\begin{definition}\label{def:laxinj}
Let $Lax_{Inj}(\ls^{op},Poset)$ denote the category with objects given
by functors from $\ls^{op}$ to $Poset$, maps given by lax
transformations that strictly respect injections, and composition
given pointwise.
\end{definition}

\begin{proposition}\label{prop:ff} cf~\cite{PP2}
Let $J:Inj^{op}\longrightarrow \ls$ be the canonical inclusion. 
Define 
\[
\int :Lax_{Inj}(\ls^{op},Poset)\longrightarrow Lax_{Inj}(\ls^{op},Poset)
\]
on objects as above.
Given $\alpha:H\Rightarrow K$, define $\int\! \alpha(n)$ 
by the fact that $j\in n/Inj$ is coherently isomorphic to the canonical inclusion $j:n\longrightarrow n+k$ for a unique natural number $k$, and applying the definition of $\int\! H(n)$ as a colimit to the cocone given by composing
\[
\alpha_{n+k}:H(m) = H(n+k)\longrightarrow K(n+k) = K(m)
\]
with the canonical map $K(m)\longrightarrow \int\! K(n)$ exhibiting $\int\! K(n)$ as a colimit.
Then $\int\!(-)$ is an endofunctor on  $Lax_{Inj}(\ls^{op},Poset)$.
\end{proposition}

The proof is routine, albeit after lengthy calculation involving
colimits. 

We can now model an arbitrary logic program by a map
$p:At\longrightarrow P_fP_f\int\! At$ in $Lax_{Inj}(\ls^{op},Poset)$,
modelling ListNat as we did in Example~\ref{ex:listnat2} but now
modelling the clauses of GC directly rather than using the awkward
substitution instances of Example~\ref{ex:lp2}.

\begin{example}\label{ex:listnat3}
Except for the restriction of $Lax(\ls^{op},Poset)$ to
$Lax_{Inj}(\ls^{op},Poset)$, ListNat is modelled in exactly the same
way here as it was in Example~\ref{ex:listnat2}, the reason being that
no clause in ListNat has a variable in the tail that does not already
appear in the head. We need only observe that, although $p$ is not
strictly natural in general, it does strictly respect injections. For
example, if one views $\mathtt{list(cons( x, 0))}$ as an element of
$At(2)$, its image under $p_2$ agrees with its image under $p_1$.
\end{example}

\begin{example}\label{ex:lp3}
In contrast to Example~\ref{ex:lp2}, using $P_fP_f\int$, we can emulate
the construction of Examples~\ref{ex:listnat2} and~\ref{ex:listnat3}
for ListNat to model GC.

Modulo possible renaming of variables, $\mathtt{connected(x,y)}$ is an
element of $At(2)$. The function $p_2$ sends it to the element $\{ \{
\mathtt{edge(x,z)},\mathtt{connected(z,y)}\}\}$ of
$(P_fP_f\int\! At)(2)$. This is possible by taking $n=2$ and $m=3$ in the
formula for $\int\! At$. In contrast,
$\{ \{ \mathtt{edge(x,z)},\mathtt{connected(z,y)}\}\}$ is not an
element of $P_fP_fAt(2)$, hence the failure of Example~\ref{ex:lp2}.

The behaviour of $P_fP_f\int\! At$ on maps ensures that the lax
transformation $p$ strictly respects injections. For example, if
$\mathtt{connected(x,y)}$ is seen as an element of $At(3)$, the
additional variable is treated as a fresh variable $w$, so does not
affect the image of $\mathtt{connected(x,y)}$ under $p_3$.
\end{example}

\begin{theorem}\label{constr:Gcoalgrefined}
The functor $P_fP_f\!\int:Lax_{Inj}(\ls^{op},Poset)\longrightarrow
Lax_{Inj}(\ls^{op},Poset)$ induces a cofree comonad $C(P_fP_f\!\int)$ on
$Lax_{Inj}(\ls^{op},Poset)$. Moreover, given a logic progam $P$ qua
$P_fP_f\!\int\!$-coalgebra $p:At\longrightarrow P_fP_f\!\int\!At$, the corresponding
$C(P_fP_f\!\int )$-coalgebra $\overline{p}:At\longrightarrow C(P_fP_f\!\int )(At)$
sends an atom $A(x_1,\ldots ,x_n)\in At(n)$ to the coinductive tree
for $A(x_1,\ldots ,x_n)$.
\end{theorem}

\begin{proof}
If one restricts $P_fP_f\!\int$ to $[Inj,Poset]$, there is a cofree comonad on it for
general reasons, $[Inj,Poset]$ being locally finitely presentable and $P_fP_f\!\int$ being
an accessible functor~\cite{W}. However, as we seek a little more generality than that, and for
completeness, we shall construct the cofree comonad. 

Observe that products in the category $Lax_{Inj}(\ls^{op},Poset)$ are given
pointwise, with pointwise projections. Moreover, those projections are strictly
natural, as one can check directly but which is also an instance of the main
result of~\cite{BKP}.

We can describe the cofree comonad $C(P_fP_f\!\int)$ on $Lax_{Inj}(\ls^{op},Poset)$  pointwise as the same limit
as in the proof of Theorem~\ref{constr:Gcoalg}, similarly to Theorem~\ref{constr:Gcoalgcount}. In particular, replacing $At$ by $At(n)$
and replacing $P_fP_f$ by $P_fP_f\!\int$ in the diagram in the proof of Theorem~\ref{constr:Gcoalg}, one has
$$\ldots \longrightarrow \At(n) \times (P_fP_f\!\int)(\At \times P_fP_f\!\int\!\At)(n)
\longrightarrow \At(n) \times (P_fP_f\!\int\!\At)(n) \longrightarrow \At(n)$$ with
maps determined by the projection $\pi_0:At(-)\times
(P_fP_f\!\int)At(-)\longrightarrow At(-)$, with the endofunctor 
$P_fP_f\!\int$ applied to it. One takes the limit, potentially transfinite~\cite{W}, of
the diagram. The limit property routinely determines
a functor $C(P_fP_f\!\int)At$.

It is routine, albeit tedious, to use the limiting property to verify functoriality of $C(P_fP_f\!\int )$ with respect to all maps, to define the counit and comultiplication, and to verify their axioms and the universal property. The construction of $\overline{p}$ is given
pointwise, with it following from its coinductive construction that it
yields the coinductive trees as required: because of our construction of $\int\! At$ to take the place $At$ in Theorem~\ref{constr:Gcoalgcount}, the image of $p$ lies in $P_fP_f\!\int\! At$.
\end{proof}

 The lax naturality in respect to general maps $f:m\longrightarrow n$
 means that a substitution applied to an atom $A(x_1,\ldots ,x_n)\in
 At(n)$, i.e., application of the function $At(f)$ to $A(x_1,\ldots
 ,x_n)$, followed by application of $\overline{p}$, i.e., taking the
 coinductive tree for the substituted atom, or application of the
 function $(C(P_fP_f\!\int )At)f)$ to the coinductive tree for $A(x_1,\ldots
 ,x_n)$ potentially yield different trees: the former substitutes into
 $A(x_1,\ldots ,x_n)$, then takes its coinductive tree, while the
 latter applies a substitution to each node of the coinductive tree
 for $A(x_1,\ldots ,x_n)$, then prunes to remove redundant branches.

\begin{example}\label{ex:lp4}
Extending Example~\ref{ex:lp3}, consider $\mathtt{connected(x,y)}\in
At(2)$. In expressing GC as a map $p:At\longrightarrow P_fP_f\!\int\!At$ in
Example~\ref{ex:lp3}, we put
\[
p_2(\mathtt{connected(x,y)}) = \{ \{
\mathtt{edge(x,z)},\mathtt{connected(z,y)}\}\}
\] 
Accordingly, $\overline{p}_2(\mathtt{connected(x,y)})$ is the
coinductive tree for $\mathtt{connected(x,y)}$, thus the infinite tree
generated by repeated application of the same clause modulo renaming
of variables.

If we substitute $x$ for $y$ in the coinductive tree, i.e., apply the
function $(C(P_fP_f\!\int )At)(x,x)$ to it (see the definition of
$L_{\Sigma}$ at the start of Section~\ref{sec:recall} and observe that
$(x,x)$ is a $2$-tuple of terms generated trivially by the variable
$x$), we obtain the same tree but with $y$ systematically replaced by
$x$. However, if we substitute $x$ for $y$ in
$\mathtt{connected(x,y)}$, i.e., apply the function $At(x,x)$ to it,
we obtain $\mathtt{connected(x,x)}\in At(1)$, whose coinductive tree
has additional branching as the first clause of GC, i.e.,
$\mathtt{connected(x,x)}\gets \,$ may also be applied.

In contrast to this, we have strict naturality with respect to
injections: for example, an injection $i:2\longrightarrow 3$ yields
the function $At(i):At(2)\longrightarrow At(3)$ that, modulo renaming
of variables, sends $\mathtt{connected(x,y)}\in At(2)$ to itself seen
as an element of $At(3)$, and the coinductive tree for
$\mathtt{connected(x,y)}$ is accordingly also sent by
$(C(P_fP_f\!\int )At)(i)$ to itself seen as an element of $(C(P_fP_f\!\int )At)(3)$.
\end{example}

Example~\ref{ex:lp4} illustrates why, although the condition of strict
naturality with respect to injections holds for $P_fP_f\!\int$, it does not
hold for $Lax(\ls^{op},P_fP_f)$ in Example~\ref{ex:lp2} as we did not
model the clause
\[
 \mathtt{connected(x,y)}  \gets  \mathtt{edge(x,z)},
 \mathtt{connected(z,y)}
\]
directly there, but rather modelled all substitution instances into
all available variables.

Turning to the relationship between lax semantics and saturated semantics given
in Section~\ref{sec:sat}, we need to refine our construction of the right adjoint to the
inclusion
\[
[\ls^{op},Poset]\longrightarrow Lax(\ls^{op},Poset)
\]
to give a construction of a right adjoint to the inclusion 
\[
[\ls^{op},Poset]\longrightarrow Lax_{Inj}(\ls^{op},Poset)
\]
As was the case in Section~\ref{sec:sat}, such a right adjoint exists for general reasons as an example of the main result of~\cite{BKP}. An explicit construction of it arises by emulating the construction of Theorem~\ref{thm:R}. In the statement of Theorem~\ref{thm:R}, putting $C = \ls^{op}$, we described a parallel pair of maps in $[\ls^{op},Poset]$ and constructed their inserter, the inserter being exactly the universal property corresponding to the laxness of the maps in $Lax(\ls^{op},Poset)$. Here, we use the same technique but with equaliser replacing inserter, to account for the equalities in $Lax_{Inj}(C,Poset)$. Thus we take an equaliser of two variants
of $\delta_1$ and $\delta_2$ seen as maps in $[\ls^{op},Poset]$ with domain $Ins(\delta_1,\delta_2)$ 

Again, the constructions of saturation and desaturation
remain the same, allowing us to maintain the relationship between lax semantics and saturated semantics. That said, our refinement of lax semantics constitutes a refinement of saturated semantics too, as, just as we now model $GC$ by a lax transformation $p:At\longrightarrow P_fP_f\!\int\!At$, one can now consider the saturation of this definition of $p$ rather than that of the less subtle map with codomain $P_cP_fAt$ used in previous papers such as~\cite{KoPS} and~\cite{BonchiZ15}.

\section{Semantics for variables in logic programs: local variables}\label{sec:local}

The relationship between the semantics of logic programming we propose here and that
of local state is yet to be explored fully, and we leave the bulk of it to future work. However,
as explained in Section~\ref{sec:derivation}, the definition of $\int$ was informed by
the semantics for local state in~\cite{PP2}, and we have preliminary results that
strengthen the relationship. 

\begin{proposition}\label{prop:monad}
The endofunctor $\int\!(-)$ on $Lax_{Inj}(\ls^{op},Poset)$ canonically supports the
structure of a monad, with unit $\eta_H:H\Rightarrow \int\! H$ defined, at $n$,  by the 
$id_n$ component  $\rho_{id_n}:Hn\longrightarrow \int\! Hn$ of the colimiting cocone,
and with multiplication $\mu_H:\int\! \int\! H \Longrightarrow \int\! H$ defined, at $n$, by
observing that if $m=n+k$ and $p=m+l$, then $p=n+(k+l)$ with canonical injections $j_k$,
$j_l$ and $j_{k+l}$
coherent with each other, and applying the doubly indexed colimiting property of $\int\! \int\! H$
to $\rho_{j_{k+l}}:H(p)\longrightarrow \int\! H(n)$,
\end{proposition}

This bears direct comparison with the monad for local state in the case where one
has only one value, as studied by Stark~\cite{Stark}. The setting is a little different.
Stark does not consider maps in $\ls$ or laxness, and his base category is $Set$ rather than
$Poset$. However, if one restricts our definition of $\int$ and the other data for the monad
of Proposition~\ref{prop:monad} to $[Inj,Set]$, one obtains Stark's construction.

The monad for local state in~\cite{PP2} also extends Stark's construction but 
in a different direction: for local state, neither $Inj$ nor $Set$ is extended, but state,
which is defined by a functor into $Set$ is interpolated into the definition of the functor $\int$,
which restricts to $[Inj,Set]$.
That interpolation of state is closely related to our application of $P_fP_f$ to $\int$: just
as the former gives rise to a monad for local state on $[Inj,Set]$, the latter bears the
ingredients for a monad as follows. 

\begin{proposition}\label{prop:monad2}
For any endofunctor $P$ on $Set$, here is a canonical 
distributive law
\[
\int\! P(-)\longrightarrow P\!\int\! (-)
\]
of the endofunctor $P\circ -$ over the monad $\int$ on $[Inj,Set]$.
\end{proposition}

The canonicity of the distributive law arises as $\int$ is defined pointwise as a 
colimit, and the distributive law is the canonical comparison map determined by applying
$P$ pointwise to the colimiting cone defining $\int$.

The functor $P_fP_f$ does not quite satisfy the axioms for a monad~\cite{Engeler} (see also~\cite{BonchiZ15}), but variants
of $P_fP_f$, in particular $P_fM_f$, where $M_f$ is the finite multiset monad on $Set$, do~\cite{Engeler} (also see~\cite{BonchiZ15}). Putting $P = P_fM_f$, the distributive law of Proposition~\ref{prop:monad2} respects the monad structure of $P_fM_f$,
yielding a canonical monad structure on the composite $P_fM_f\int$. 

The full implications
of that are yet to be investigated, but, trying to emulate the analysis of local state in~\cite{PP2}, we believe we have a natural set of operations and equations that generate the monad $P_fM_f\int$. That encourages us considerably towards the possibility of seeing the semantics for logic programming, both lax and saturated, as an example of a general semantics of  local effects. We have not yet fully understood the significance of the
specific combination of operations and equations generating the monad, but we are currently
investigating it.

 \section{Conclusions and Further Work}\label{sec:concl}
Let $P_f$ be the covariant finite powerset functor on $Set$. Then,  to give a variable-free logic program $P$ is equivalent to giving a $P_fP_f$-coalgebra structure $p:At\longrightarrow P_fP_fAt$ on the set $At$ of atoms in the program. 
Now let $C(P_fP_f)$ be the cofree comonad on $P_fP_f$. Then, the $C(P_fP_f)$-coalgebra $\overline{p}:At\longrightarrow C(P_fP_f)At$ corresponding to $p$ sends an atom to the coinductive tree it generates.  This fact is the basis for both our lax semantics and
Bonchi and Zanasi's saturated semantics for logic programming. 

Two problems arise when, following standard category theoretic practice, one tries to extend this semantics to model logic programs in general by extending from $Set$ to $[\ls^{op},Set]$, where $\ls$ is the free Lawevere theory generated by a signature $\Sigma$. The first is that the natural
construction $p:At\longrightarrow P_fP_fAt$ does not form a natural transformation, so is not a map in $[\ls^{op},Set]$. 

Two resolutions were proposed to that: lax semantics~\cite{KoPS}, which we have been developing in the tradition of semantics for data refinement~\cite{HH}, and saturated semantics~\cite{BonchiZ15}, which Bonchi and Zanasi have been developing. In this paper, we have shown that the two resolutions are complementary rather than competing, the first modelling the theorem-proving aspect of logic programming, while the latter models proof search.

In modelling theorem-proving, lax semantics led us to identify and develop the notion of coinductive tree. To express the semantics, we extended $[\ls^{op},Set]$ to $Lax(\ls^{op},Poset)$, the category of strict functors and lax transformations between them. We followed standard semantic practice in extended $P_f$ from $Set$ to $Poset$ and we postcomposed the functor $At:\ls^{op}\longrightarrow Poset$ by $P_fP_f$. Bonchi
and Zanasi also postcomposed $At$ by $P_fP_f$, but then saturated. We showed that their
saturation and desaturation constructions are generated exactly by starting from $Lax(\ls^{op},Poset)$ rather than from $[ob(\ls)^{op},Set]$ as they did,
thus unifying the underlying mathematics of the two developments, supporting their
computational coherence. 

The second problem mentioned above relates to existential variables, those being variables that appear in the antecedent of a clause but not in its head. The problem of existential clauses is well-known
in the literature on theorem proving and within communities that use
term-rewriting, TM-resolution or their variants.  In
TRS~\cite{Terese}, existential variables are not allowed to appear in
rewriting rules, and in type inference, the restriction to
non-existential programs is common~\cite{Jones97}. In logic programming, the problem
of handling existential variables when constructing proofs with
TM-resolution marks the boundary between its theorem-proving and
problem-solving aspects.

Existential variables are not present in many logic programs, but they do occasionally occur in important examples, such as those developed by Sterling and Shapiro~\cite{SS}. The problem for us was that, in the presence of existential variables, the natural model $p:At\longrightarrow P_fP_fAt$ of a logic program might escape its codomain, i.e., $p_n(A)$ might not lie in $P_fP_fAt(n)$ because of the new variables. On one hand, we want to model them, but on the other, the fact of the difficulty for us means that we have semantically identified the concept of existential variable, which is positive.

In this paper, we have resolved the problem by refining $Lax(\ls^{op},Poset)$ to
$Lax_{Inj}(\ls^{op},Poset)$, insisting upon strict naturality for
injections, and by refining the construction $P_fP_fAt$ to $P_fP_f\!\int\! At$,
thus allowing for additional variables in the tail of a clause in a
logic program. That has allowed us to model coinductive
trees for arbitrary logic programs, in particular those including existential variables.
We have also considered
the effect of such refinement on saturated semantics. 

In order to refine $P_fP_f(-)$, we followed a
technique developed in the semantics of local state~\cite{PP2}. That alerted us
to the relationship between variables in logic programming with the use of worlds
in modelling local state. So, as ongoing work, we are now relating our semantics for logic programming with that for local state. For the future, we shall continue to develop that, with the hope of being able to locate our semantics of logic programming within a general semantics for local effects.

Beyond that, a question that we have not considered semantically at all yet but which our applied investigations are encouraging is that of modelling recursion. There are fundamentally two kinds of recursion that arise in logic programming as there may be recursion in terms and recursion in proofs.
For example, $\mathtt{stream(scons(x,y)) \gets stream(y)}$ is a standard  (co-)recursive definition of infinite streams in logic programming literature. More abstractly, the following program $P_1$: $\mathtt{p(f(x)) \gets p(x)}$ defines an infinite data structure $p$ with constructor $f$. For such cases,  proofs given by coinductive trees will be finite. An infinite sequence of (finite) coinductive trees will be needed to approximate the intended operational semantics of such a program, as we discuss in detail in~\cite{KoPS,KJ15}. In contrast, there are programs like $P_2$: $\mathtt{p(x) \gets p(x)}$ or $P_3$: $\mathtt{p(x) \gets p(f(x))}$  that are also recursive, but additionally their proofs as given by coinductive trees will have an infinite size. 

In~\cite{KJS16,FK16} programs like $P_1$ were called productive, for producing (infinite) data, and programs like $P_2$ and $P_3$ -- non-productive, for recursing without producing substitutions.
 The productive case amounts to a loop in the Lawvere theory $\ls$, while the non-productive case amounts to repetition within a coinductive tree, possibly modulo a substitution. This paper gives a close analysis of trees. That should set the scene for investigation of recursion, as it seems likely to yield more general kinds of graph that arise by identifying loops in $\ls$ and by equating nodes in trees. 

The lax semantics we presented here has recently inspired investigations into importance of TM-resolution (as modelled by the coinductive trees) in programming languages. In particular, TM-resolution is used in type class inference in Haskell~\cite{Jones97}. In~\cite{FKSP16,FKHF16} we showed applications of nonterminating TM-resolution in Haskell type classes. We plan to continue looking for applications of this work in programming language design beyond logic programming.



 \section*{Acknowledgements}
Ekaterina Komendantskaya would like to acknowledge the support of
EPSRC Grant EP/K031864/1-2.  John Power would like to acknowledge the
support of EPSRC grant EP/K028243/1 and Royal Society grant IE151369

\section*{Bibliography}
  
	\bibliographystyle{elsarticle-num-names} 
  \bibliography{CMCS}

\begin{thebibliography}{39}
\providecommand{\natexlab}[1]{#1}
\providecommand{\url}[1]{\texttt{#1}}
\providecommand{\urlprefix}{URL }
\expandafter\ifx\csname urlstyle\endcsname\relax
  \providecommand{\doi}[1]{doi:\discretionary{}{}{}#1}\else
  \providecommand{\doi}[1]{doi:\discretionary{}{}{}\begingroup
  \urlstyle{rm}\url{#1}\endgroup}\fi
\providecommand{\bibinfo}[2]{#2}

\bibitem[{Komendantskaya et~al.(2016{\natexlab{a}})Komendantskaya, Power, and
  Schmidt}]{KoPS}
\bibinfo{author}{E.~Komendantskaya}, \bibinfo{author}{J.~Power},
  \bibinfo{author}{M.~Schmidt}, \bibinfo{title}{Coalgebraic logic programming:
  from Semantics to Implementation}, \bibinfo{journal}{J. Log. Comput.}
  \bibinfo{volume}{26}~(\bibinfo{number}{2})
  (\bibinfo{year}{2016}{\natexlab{a}}) \bibinfo{pages}{745 -- 783}.

\bibitem[{Bonchi and Zanasi(2015)}]{BonchiZ15}
\bibinfo{author}{F.~Bonchi}, \bibinfo{author}{F.~Zanasi},
  \bibinfo{title}{Bialgebraic Semantics for Logic Programming},
  \bibinfo{journal}{Logical Methods in Computer Science}
  \bibinfo{volume}{11}~(\bibinfo{number}{1}).

\bibitem[{Kelly(1982)}]{K}
\bibinfo{author}{G.~Kelly}, \bibinfo{title}{Basic concepts of enriched category
  theory}, \bibinfo{journal}{London Math. Soc. Lecture Notes Series}
  \bibinfo{volume}{64}.

\bibitem[{Gupta and Costa(1994)}]{GC}
\bibinfo{author}{G.~Gupta}, \bibinfo{author}{V.~Costa}, \bibinfo{title}{Optimal
  implementation of and-or parallel Prolog}, in: \bibinfo{booktitle}{PARLE'92},
  \bibinfo{pages}{71 -- 92}, \bibinfo{year}{1994}.

\bibitem[{Bruni et~al.(2001)Bruni, Montanari, and Rossi}]{BMR}
\bibinfo{author}{R.~Bruni}, \bibinfo{author}{U.~Montanari},
  \bibinfo{author}{F.~Rossi}, \bibinfo{title}{An interactive semantics of logic
  programming}, \bibinfo{journal}{{TPLP}}
  \bibinfo{volume}{1}~(\bibinfo{number}{6}) (\bibinfo{year}{2001})
  \bibinfo{pages}{647 -- 690}.

\bibitem[{Bonchi and Montanari(2009)}]{BM}
\bibinfo{author}{F.~Bonchi}, \bibinfo{author}{U.~Montanari},
  \bibinfo{title}{Reactive systems, (semi-)saturated semantics and coalgebras
  on presheaves}, \bibinfo{journal}{Theor. Comput. Sci.}
  \bibinfo{volume}{410}~(\bibinfo{number}{41}) (\bibinfo{year}{2009})
  \bibinfo{pages}{4044 -- 4066}.

\bibitem[{Comini et~al.(2001)Comini, Levi, and Meo}]{CLM}
\bibinfo{author}{M.~Comini}, \bibinfo{author}{G.~Levi}, \bibinfo{author}{M.~C.
  Meo}, \bibinfo{title}{A Theory of Observables for Logic Programs},
  \bibinfo{journal}{Inf. Comput.} \bibinfo{volume}{169}~(\bibinfo{number}{1})
  (\bibinfo{year}{2001}) \bibinfo{pages}{23 -- 80}.

\bibitem[{Lloyd(1987)}]{Llo}
\bibinfo{author}{J.~Lloyd}, \bibinfo{title}{Foundations of Logic Programming},
  \bibinfo{publisher}{Springer-Verlag}, \bibinfo{year}{1987}.

\bibitem[{Komendantskaya et~al.(2011)Komendantskaya, McCusker, and Power}]{KMP}
\bibinfo{author}{E.~Komendantskaya}, \bibinfo{author}{G.~McCusker},
  \bibinfo{author}{J.~Power}, \bibinfo{title}{Coalgebraic semantics for
  parallel derivation strategies in logic programming}, in:
  \bibinfo{booktitle}{AMAST'2010}, vol. \bibinfo{volume}{6486} of
  \emph{\bibinfo{series}{Lecture Notes in Computer Science}},
  \bibinfo{publisher}{Springer}, \bibinfo{pages}{111 -- 127},
  \bibinfo{year}{2011}.

\bibitem[{Komendantskaya et~al.(2014)Komendantskaya, Schmidt, and
  Heras}]{KSH14}
\bibinfo{author}{E.~Komendantskaya}, \bibinfo{author}{M.~Schmidt},
  \bibinfo{author}{J.~Heras}, \bibinfo{title}{Exploiting Parallelism in
  Coalgebraic Logic Programming}, \bibinfo{journal}{Electr. Notes Theor.
  Comput. Sci.} \bibinfo{volume}{303} (\bibinfo{year}{2014})
  \bibinfo{pages}{121 -- 148}.

\bibitem[{Johann et~al.(2015)Johann, Komendantskaya, and
  Komendantskiy}]{JohannKK15}
\bibinfo{author}{P.~Johann}, \bibinfo{author}{E.~Komendantskaya},
  \bibinfo{author}{V.~Komendantskiy}, \bibinfo{title}{Structural Resolution for
  Logic Programming}, in: \bibinfo{booktitle}{Proceedings of the Technical
  Communications of the 31st International Conference on Logic Programming
  {(ICLP} 2015), Cork, Ireland, August 31 - September 4, 2015.}, vol.
  \bibinfo{volume}{1433} of \emph{\bibinfo{series}{{CEUR} Workshop
  Proceedings}}, \bibinfo{year}{2015}.

\bibitem[{Fu and Komendantskaya(2015)}]{FK15}
\bibinfo{author}{P.~Fu}, \bibinfo{author}{E.~Komendantskaya}, \bibinfo{title}{A
  Type-Theoretic Approach to Resolution}, in: \bibinfo{booktitle}{Logic-Based
  Program Synthesis and Transformation - 25th International Symposium, {LOPSTR}
  2015, Siena, Italy, July 13-15, 2015. Revised Selected Papers}, vol.
  \bibinfo{volume}{9527} of \emph{\bibinfo{series}{Lecture Notes in Computer
  Science}}, \bibinfo{publisher}{Springer}, \bibinfo{pages}{91 -- 106},
  \bibinfo{year}{2015}.

\bibitem[{Fu et~al.(2016)Fu, Komendantskaya, Schrijvers, and Pond}]{FKSP16}
\bibinfo{author}{P.~Fu}, \bibinfo{author}{E.~Komendantskaya},
  \bibinfo{author}{T.~Schrijvers}, \bibinfo{author}{A.~Pond},
  \bibinfo{title}{Proof Relevant Corecursive Resolution}, in:
  \bibinfo{booktitle}{Functional and Logic Programming - 13th International
  Symposium, {FLOPS} 2016, Kochi, Japan, March 4-6, 2016, Proceedings}, vol.
  \bibinfo{volume}{9613} of \emph{\bibinfo{series}{Lecture Notes in Computer
  Science}}, \bibinfo{publisher}{Springer}, \bibinfo{pages}{126 -- 143},
  \bibinfo{year}{2016}.

\bibitem[{Bonchi and Zanasi(2013)}]{BZ}
\bibinfo{author}{F.~Bonchi}, \bibinfo{author}{F.~Zanasi},
  \bibinfo{title}{Saturated Semantics for Coalgebraic Logic Programming}, in:
  \bibinfo{booktitle}{Algebra and Coalgebra in Computer Science - 5th
  International Conference, {CALCO} 2013, Warsaw, Poland, September 3-6, 2013.
  Proceedings}, vol. \bibinfo{volume}{8089} of \emph{\bibinfo{series}{Lecture
  Notes in Computer Science}}, \bibinfo{publisher}{Springer},
  \bibinfo{pages}{80 -- 94}, \bibinfo{year}{2013}.

\bibitem[{Amato et~al.(2009)Amato, Lipton, and McGrail}]{ALM}
\bibinfo{author}{G.~Amato}, \bibinfo{author}{J.~Lipton},
  \bibinfo{author}{R.~McGrail}, \bibinfo{title}{On the algebraic structure of
  declarative programming languages}, \bibinfo{journal}{Theor. Comput. Sci.}
  \bibinfo{volume}{410}~(\bibinfo{number}{46}) (\bibinfo{year}{2009})
  \bibinfo{pages}{4626 -- 4671}.

\bibitem[{Kinoshita and Power(1996{\natexlab{a}})}]{KP96}
\bibinfo{author}{Y.~Kinoshita}, \bibinfo{author}{A.~J. Power},
  \bibinfo{title}{A Fibrational Semantics for Logic Programs}, in:
  \bibinfo{booktitle}{Extensions of Logic Programming, 5th International
  Workshop, ELP'96, Leipzig, Germany, March 28-30, 1996, Proceedings}, vol.
  \bibinfo{volume}{1050} of \emph{\bibinfo{series}{Lecture Notes in Computer
  Science}}, \bibinfo{publisher}{Springer}, \bibinfo{pages}{177 -- 191},
  \bibinfo{year}{1996}{\natexlab{a}}.

\bibitem[{Komendantskaya and Power(2011{\natexlab{a}})}]{KoP}
\bibinfo{author}{E.~Komendantskaya}, \bibinfo{author}{J.~Power},
  \bibinfo{title}{Coalgebraic Semantics for Derivations in Logic Programming},
  in: \bibinfo{booktitle}{Algebra and Coalgebra in Computer Science - 4th
  International Conference, {CALCO} 2011, Winchester, UK, August 30 - September
  2, 2011. Proceedings}, vol. \bibinfo{volume}{6859} of
  \emph{\bibinfo{series}{Lecture Notes in Computer Science}},
  \bibinfo{publisher}{Springer}, \bibinfo{pages}{268 -- 282},
  \bibinfo{year}{2011}{\natexlab{a}}.

\bibitem[{Benabou(1967)}]{Ben}
\bibinfo{author}{J.~Benabou}, \bibinfo{title}{Introduction to bicategories},
  \bibinfo{journal}{Lecture Notes in Mathematics} \bibinfo{volume}{47}
  (\bibinfo{year}{1967}) \bibinfo{pages}{1 -- 77}.

\bibitem[{Blackwell et~al.(1989)Blackwell, Kelly, and Power}]{BKP}
\bibinfo{author}{R.~Blackwell}, \bibinfo{author}{G.~Kelly},
  \bibinfo{author}{A.~Power}, \bibinfo{title}{Two-dimensional monad theory},
  \bibinfo{journal}{J. Pure Appl. Algebra} \bibinfo{volume}{59}
  (\bibinfo{year}{1989}) \bibinfo{pages}{1 -- 41}.

\bibitem[{Kelly(1974)}]{Kelly}
\bibinfo{author}{G.~Kelly}, \bibinfo{title}{Coherence theorems for lax algebras
  and for distributive laws}, in: \bibinfo{booktitle}{Category seminar}, vol.
  \bibinfo{volume}{420} of \emph{\bibinfo{series}{Lecture Notes in
  Mathematics}}, \bibinfo{publisher}{Spriniger}, \bibinfo{pages}{281 -- 375},
  \bibinfo{year}{1974}.

\bibitem[{He and Hoare(1989)}]{HH}
\bibinfo{author}{J.~He}, \bibinfo{author}{C.~A.~R. Hoare},
  \bibinfo{title}{Categorical Semantics for Programming Languages}, in:
  \bibinfo{booktitle}{Mathematical Foundations of Programming Semantics, 5th
  International Conference, Tulane University, New Orleans, Louisiana, USA,
  March 29 - April 1, 1989, Proceedings}, vol. \bibinfo{volume}{442} of
  \emph{\bibinfo{series}{Lecture Notes in Computer Science}},
  \bibinfo{publisher}{Springer}, \bibinfo{pages}{402 -- 417},
  \bibinfo{year}{1989}.

\bibitem[{Jifeng and Hoare(1990)}]{HH1}
\bibinfo{author}{H.~Jifeng}, \bibinfo{author}{C.~Hoare}, \bibinfo{title}{Data
  refinement in a categorical setting}, \bibinfo{type}{Tech. Rep.}
  \bibinfo{number}{Technical Monograph PRG-90}, \bibinfo{institution}{Oxford
  University Computing Laboratory, Programming Research Group, Oxford},
  \bibinfo{year}{1990}.

\bibitem[{Kinoshita and Power(1996{\natexlab{b}})}]{KP1}
\bibinfo{author}{Y.~Kinoshita}, \bibinfo{author}{A.~Power}, \bibinfo{title}{Lax
  naturality through enrichment}, \bibinfo{journal}{J. Pure Appl. Algebra}
  \bibinfo{volume}{112} (\bibinfo{year}{1996}{\natexlab{b}}) \bibinfo{pages}{53
  -- 72}.

\bibitem[{Power(1989)}]{P}
\bibinfo{author}{A.~J. Power}, \bibinfo{title}{An Algebraic Formulation for
  Data Refinement}, in: \bibinfo{booktitle}{Mathematical Foundations of
  Programming Semantics, 5th International Conference, Tulane University, New
  Orleans, Louisiana, USA, March 29 - April 1, 1989, Proceedings}, vol.
  \bibinfo{volume}{442} of \emph{\bibinfo{series}{Lecture Notes in Computer
  Science}}, \bibinfo{publisher}{Springer}, \bibinfo{pages}{390 -- 401},
  \bibinfo{year}{1989}.

\bibitem[{Komendantskaya and Power(2016)}]{KoP1}
\bibinfo{author}{E.~Komendantskaya}, \bibinfo{author}{J.~Power},
  \bibinfo{title}{Category Theoretic Semantics for Theorem Proving in Logic
  Programming: Embracing the Laxness}, in: \bibinfo{booktitle}{Coalgebraic
  Methods in Computer Science - 13th {IFIP} {WG} 1.3 International Workshop,
  {CMCS} 2016, Colocated with {ETAPS} 2016, Eindhoven, The Netherlands, April
  2-3, 2016, Revised Selected Papers}, vol. \bibinfo{volume}{9608} of
  \emph{\bibinfo{series}{Lecture Notes in Computer Science}},
  \bibinfo{publisher}{Springer}, \bibinfo{pages}{94 -- 113},
  \bibinfo{year}{2016}.

\bibitem[{Sterling and Shapiro(1986)}]{SS}
\bibinfo{author}{L.~Sterling}, \bibinfo{author}{E.~Shapiro},
  \bibinfo{title}{The art of Prolog}, \bibinfo{publisher}{MIT Press},
  \bibinfo{year}{1986}.

\bibitem[{Komendantskaya and Power(2011{\natexlab{b}})}]{KoP2}
\bibinfo{author}{E.~Komendantskaya}, \bibinfo{author}{J.~Power},
  \bibinfo{title}{Coalgebraic Derivations in Logic Programming}, in:
  \bibinfo{booktitle}{Computer Science Logic, 25th International Workshop /
  20th Annual Conference of the EACSL, {CSL} 2011, September 12-15, 2011,
  Bergen, Norway, Proceedings}, vol.~\bibinfo{volume}{12} of
  \emph{\bibinfo{series}{LIPIcs}}, \bibinfo{publisher}{Schloss Dagstuhl -
  Leibniz-Zentrum fuer Informatik}, \bibinfo{pages}{352 -- 366},
  \bibinfo{year}{2011}{\natexlab{b}}.

\bibitem[{Plotkin and Power(2002)}]{PP2}
\bibinfo{author}{G.~D. Plotkin}, \bibinfo{author}{J.~Power},
  \bibinfo{title}{Notions of Computation Determine Monads}, in:
  \bibinfo{booktitle}{Foundations of Software Science and Computation
  Structures, 5th International Conference, {FOSSACS} 2002. Held as Part of the
  Joint European Conferences on Theory and Practice of Software, {ETAPS} 2002
  Grenoble, France, April 8-12, 2002, Proceedings}, vol. \bibinfo{volume}{2303}
  of \emph{\bibinfo{series}{Lecture Notes in Computer Science}},
  \bibinfo{publisher}{Springer}, \bibinfo{pages}{342 -- 356},
  \bibinfo{year}{2002}.

\bibitem[{Terese(2003)}]{Terese}
\bibinfo{author}{Terese}, \bibinfo{title}{Term Rewriting Systems},
  \bibinfo{publisher}{Cambridge University Press}, \bibinfo{year}{2003}.

\bibitem[{Simon et~al.(2007)Simon, Bansal, Mallya, and Gupta}]{SimonBMG07}
\bibinfo{author}{L.~Simon}, \bibinfo{author}{A.~Bansal},
  \bibinfo{author}{A.~Mallya}, \bibinfo{author}{G.~Gupta},
  \bibinfo{title}{Co-Logic Programming: Extending Logic Programming with
  Coinduction}, in: \bibinfo{booktitle}{Automata, Languages and Programming,
  34th International Colloquium, {ICALP} 2007, Wroclaw, Poland, July 9-13,
  2007, Proceedings}, vol. \bibinfo{volume}{4596} of
  \emph{\bibinfo{series}{Lecture Notes in Computer Science}},
  \bibinfo{publisher}{Springer}, \bibinfo{pages}{472 -- 483},
  \bibinfo{year}{2007}.

\bibitem[{Jones et~al.(1997)Jones, Jones, and Meijer}]{Jones97}
\bibinfo{author}{S.~P. Jones}, \bibinfo{author}{M.~Jones},
  \bibinfo{author}{E.~Meijer}, \bibinfo{title}{Type classes: An exploration of
  the design space}, in: \bibinfo{booktitle}{Haskell Workshop},
  \bibinfo{year}{1997}.

\bibitem[{Worrell(1999)}]{W}
\bibinfo{author}{J.~Worrell}, \bibinfo{title}{Terminal sequences for accessible
  endofunctors}, \bibinfo{journal}{Electr. Notes Theor. Comput. Sci.}
  \bibinfo{volume}{19} (\bibinfo{year}{1999}) \bibinfo{pages}{24 -- 38}.

\bibitem[{Lane(1971)}]{Mac}
\bibinfo{author}{S.~M. Lane}, \bibinfo{title}{Categories for the Working
  Mathematician}, Graduate Texts in Mathematics, \bibinfo{publisher}{Springer},
  \bibinfo{year}{1971}.

\bibitem[{Stark(1996)}]{Stark}
\bibinfo{author}{I.~Stark}, \bibinfo{title}{Categorical Models for Local
  Names}, \bibinfo{journal}{Lisp and Symbolic Computation}
  \bibinfo{volume}{9}~(\bibinfo{number}{1}) (\bibinfo{year}{1996})
  \bibinfo{pages}{77 -- 107}.

\bibitem[{Hyland et~al.(2006)Hyland, Nagayama, Power, and Rosolini}]{Engeler}
\bibinfo{author}{M.~Hyland}, \bibinfo{author}{M.~Nagayama},
  \bibinfo{author}{J.~Power}, \bibinfo{author}{G.~Rosolini}, \bibinfo{title}{A
  Category Theoretic Formulation for Engeler-style Models of the Untyped
  lambda}, \bibinfo{journal}{Electr. Notes Theor. Comput. Sci.}
  \bibinfo{volume}{161} (\bibinfo{year}{2006}) \bibinfo{pages}{43 -- 57}.

\bibitem[{Johann and Komendantskaya(2015)}]{KJ15}
\bibinfo{author}{P.~Johann}, \bibinfo{author}{E.~Komendantskaya},
  \bibinfo{title}{Structural Resolution: a Framework for Coinductive Proof
  Search and Proof Construction in Horn Clause Logic}, \bibinfo{journal}{Under
  Review, arXiv:1511.07865} .

\bibitem[{Komendantskaya et~al.(2016{\natexlab{b}})Komendantskaya, Johann, and
  Schmidt}]{KJS16}
\bibinfo{author}{E.~Komendantskaya}, \bibinfo{author}{P.~Johann},
  \bibinfo{author}{M.~Schmidt}, \bibinfo{title}{A productivity checker for
  logic programming}, in: \bibinfo{booktitle}{Accepted to pre-proceedings of
  LOPSTR'2016}, \bibinfo{note}{arXiv:1608.04415},
  \bibinfo{year}{2016}{\natexlab{b}}.

\bibitem[{Fu and Komendantskaya(2016)}]{FK16}
\bibinfo{author}{P.~Fu}, \bibinfo{author}{E.~Komendantskaya},
  \bibinfo{title}{Operational semantics of resolution and productivity in
  {H}orn clause logic}, \bibinfo{journal}{Formal Aspects of Computing}
  (\bibinfo{year}{2016})
  \bibinfo{pages}{Accepted,}\bibinfo{note}{ArXiv:1604.04114}.

\bibitem[{Farka et~al.(2016)Farka, Komendantskaya, Hammond, and Fu}]{FKHF16}
\bibinfo{author}{F.~Farka}, \bibinfo{author}{E.~Komendantskaya},
  \bibinfo{author}{K.~Hammond}, \bibinfo{author}{P.~Fu},
  \bibinfo{title}{Coinductive Soundness of Corecursive Type Class Resolution},
  in: \bibinfo{booktitle}{Accepted to pre-proceedings of LOPSTR'2016},
  \bibinfo{note}{arXiv:1608.05233}, \bibinfo{year}{2016}.

\end{thebibliography}

\end{document}